\documentclass[a4paper,12pt]{article}
\usepackage{amsfonts,amsthm,amsmath,amssymb,graphicx,relsize}
\usepackage{graphicx,epsfig,color}
\usepackage{hyperref}
\usepackage{youngtab,ytableau,simplewick}
\usepackage{enumitem}

\numberwithin{equation}{section}

\let\OLDthebibliography\thebibliography
\renewcommand\thebibliography[1]{
  \OLDthebibliography{#1}
  \setlength{\parskip}{2pt}
  \setlength{\itemsep}{2pt plus 2pt}
}

\def\vev#1{\left\langle #1 \right\rangle}
\def\rep#1{{\tt [#1]}}
\def\bb#1{\mathbb{#1}}

\def\cC{\mathcal C}
\def\cN{\mathcal N}
\def\cO{\mathcal O}
\def\cS{\mathcal S}

\def\tr{\mathrm{tr}}
\def\({\left(}
\def\){\right)}
\def\[{\left[}
\def\]{\right]}

\def\abs#1{{\left| #1\right|}}

\def\vev#1{\langle #1\rangle}
\def\ket#1{\left| #1\right\rangle}
\def\ol#1{\overline{#1}}
\newcommand{\Clebsch}[7]{S^{#1 ,}{}^{#2}_{#5}{}^{#3}_{#6}{}^{#4}_{#7}}

\newtheorem{thm}{Theorem}
\newtheorem{cor}[thm]{Corollary}

\newtheorem{lem}[thm]{Lemma}

\newcommand{\bea}{\begin{eqnarray}} 
\newcommand{\eea}{\end{eqnarray}} 
 
\newcommand\cM{ \mathcal{M}} 
\newcommand\mC{ \mathbb{C} }

\begin{document}
\thispagestyle{empty}

\rightline{OIQP-16-03}
\rightline{QMUL-PH-16-15}

\vskip 20mm

{\LARGE 
\centerline{\bf Flavour singlets in gauge theory as Permutations}
}

\vskip 12mm

\centerline{
{\large \bf Yusuke Kimura ${}^{a,}$\footnote{ {\tt londonmileend\_at\_gmail.com}}},
{\large \bf Sanjaye Ramgoolam ${}^{b,c,}$\footnote{ {\tt s.ramgoolam\_at\_qmul.ac.uk}}}
{\large \bf and Ryo Suzuki ${}^{d,}$\footnote{ {\tt rsuzuki.mp\_at\_gmail.com}}}
}

\vskip 12mm

\centerline{{\it ${}^a$ Okayama Institute for Quantum Physics (OIQP)},}
\centerline{{\it Furugyo-cho 1-7-36, Naka-ku, Okayama, 703-8278, Japan}}

\vspace{.4cm}
\centerline{{\it ${}^b$ Centre for Research in String Theory, School of Physics and Astronomy},}
\centerline{{\it Queen Mary University of London,}} 
\centerline{{\it Mile End Road, London E1 4NS, UK}}
\vspace{.2cm}
\centerline{{\it ${}^c$ National Institute for Theoretical Physics, }}
\centerline{{\it School of Physics and Mandelstam Institute for Theoretical Physics, }} 
\centerline{{\it University of Witwatersrand, Wits, 2050, South Africa}}

\vspace{.4cm}
\centerline{{\it ${}^d$ ICTP South American Institute for Fundamental Research},}
\centerline{{\it Instituto de F\'isica Te\'orica, UNESP - Universidade Estadual Paulista},}
\centerline{{\it Rua Dr. Bento Teobaldo Ferraz 271, 01140-070, S\~ao Paulo, SP, Brazil}}

\vskip 5mm


\centerline{\bf ABSTRACT}

\vskip 2mm 

Gauge-invariant operators can be specified by
equivalence classes of permutations. 
We develop this idea concretely for the singlets of the flavour group $SO(N_f)$ in $U(N_c)$ gauge theory by using Gelfand pairs and Schur-Weyl duality.
The singlet operators, when specialised at $N_f =6$, belong to the scalar sector of ${\cal N}=4$ SYM. 
A simple formula is given for the two-point functions in the free field limit of $g_{YM}^2 =0$. 
The free two-point functions are shown to be equal to the partition function on a 2-complex with boundaries and a defect, in a topological field theory of permutations.
The permutation equivalence classes are Fourier transformed to a representation basis which is   orthogonal for the two-point functions at finite $N_c , N_f$. Counting formulae for the gauge-invariant operators are described. 
The one-loop mixing matrix is derived as a linear operator on the permutation equivalence classes.

\newpage

\tableofcontents

\setcounter{page}{0}
\setcounter{tocdepth}{2}
\setcounter{footnote}{0}

\section{Introduction}\label{sec:intro}

The AdS/CFT \cite{malda,gkp,witten} correspondence has led to detailed studies of local operators in gauge theories.
A remarkable success has been the discovery of integrability in planar $\cN=4$ SYM, which allows the computation of conformal dimensions and other quantities at any value of the 't Hooft coupling \cite{Beisert-Review}. 
In parallel, the study of AdS/CFT at finite $N_c$, the rank of the gauge group, is making progress. 
The CFT duals of giant gravitons \cite{mst} are local BPS operators, which have been investigated using techniques based on permutation groups and  Fourier transformation in representation theory \cite{CJR01,CR02,KR1,BHR07,BCD08,BDS08,BHR08,KR2}. A class of quarter BPS operators, showing 
finite $N_c$ cutoffs related to Brauer algebras,  were constructed  \cite{YK-quarter}.  The general  quarter BPS
for  the case where the dimension of the operators is less than $N_c$ were constructed using permutation 
algebras in \cite{BHR07,BHR08,Brown2010,CountConst}. 
Integrability of excitations around large half-BPS operators has been established \cite{1012.3884,CdML11,1108.2761}.  A natural direction of investigation is the application of the permutation-based methods to problems in the $1/N_c$ expansion and far from planarity. In this paper we will take a further step in this direction, studying singlet operators in the $SO(6)$ sector made from six scalars. The study of non-planar corrections in this sector was initiated in \cite{KS15}.

We will focus on the sector of hermitian scalar fields in $ {\mathcal{N}} = 4 $ 
SYM with $U(N_c)$ gauge group, and consider ``mesonic'' gauge-invariant operators.
The mesonic $SO(6)$-invariant operators form a simple closed subsector under the action of one-loop dilatation operator; namely the explicit form of the one-loop dilatation tells that there is no mixing between mesonic operators and other operators like non-scalar singlets with derivatives or ``baryonic'' operators.\footnote{The absence of mixing with baryonic operators is discussed around \eqref{def:baryonic} in Section \ref{sec: conclusion}.}
It is in fact convenient to replace $SO(6)$ by $SO(N_f)$ and discuss the general $N_f$ theory.
Various arguments simplify at large $N_f$ thanks to the absence of finite $N_f$ constraints.
We obtain identities relating different ways of counting gauge-invariant operators.

There is an essential distinction between $N_c \ge 2n$ and $N_c < 2n$, where $2n$ is the operator length.
For example the planar limit is a limit in the former regime.  
The latter regime contains interesting limits like $N_c \gg 1$, $n\sim O(N_c)$, which is related to the description of giant gravitons. 
Conventionally the former is called large $N_c$ and the latter is called finite $N_c$.
Likewise we mean large $N_f$ and finite $N_f$ by $N_f \ge 2n$ and $N_f < 2n$.

In section \ref{sec: operators} we show that operators  $ \cO_{ \alpha } $ of length $2n$ can be parametrised by permutations $ \alpha $  in $ S_{ 2n}$, the symmetric group of all permutations of $2n$ distinct objects. Different permutations giving rise to the same mesonic operator are related by conjugation with an element $\gamma $ in the wreath product subgroup $S_{n} [ S_2]  \subset S_{ 2n } $. This group, of dimension $2^n n!$, contains $n$ copies of $S_2$ as well as the symmetric group $S_n$ consisting of permutations of $n$ pairs. 
The two-point function of operators $ \langle \cO_{ \alpha_1} \cO_{ \alpha_2}  \rangle $ is expressed in terms of  a sum of permutations in $S_{2n}$. The all-orders expansion in $N_f $ and $N_c$ is given in terms of symmetric group data such as cycle types of appropriate permutations. 
This wreath product group is also used in \cite{CDD1301,CDD1303} to organise the colour structure in $SO(N_c)$ and $ Sp(N_c)$ 
 gauge theories.

In section \ref{sec: rep basis}, we show how to take a linear combination of the operators labelled by permutations to form a new basis of operators labelled by  representations of $S_{2n}$, along with a group-theoretical multiplicity label. 
This procedure can be thought of as the Fourier transform on finite groups, which replaces permutation labels by representation labels. The two-point function for these representation-labelled operators is diagonal, as given in (\ref{two-pt_rep_basis}). 
The diagonal operators become null when the constraints of finite $N_c$ or $N_f$ are violated. The construction in this section is an explicit realisation for the mesonic sector of a general construction in \cite{BHR08}.

In section \ref{sec: counting} we take a careful look at the counting of the mesonic operators constructed. 
We obtain the exact counting formula using representation theory, in particular Schur-Weyl duality. The representation labels are equipped with finite $N_c$ and finite $N_f$ cut-offs, which give the correct counting for general $n$. 
We return to the language of permutations, and analyze the large $N_c,N_f$ limits.
In this way, we make contact with counting of the graphs via the Burnside Lemma.

In section \ref{sec:TFT} we consider a topological lattice gauge theory with a discrete gauge group $S_{2n}$ of the type discussed in \cite{DW,FHK,CFS}. The counting of mesonic operators in the permutation basis, and their   free-field  two-point functions  can be interpreted as the partition functions of this topological field theory, defined on 2-complexes with boundaries and a defect.

Finally in section \ref{sec:one-loop}, we derive formulae for the action of 
the one-loop dilatation operator acting on the operators constructed above.

In appendices we will explain our notation, collect mathematical statements and details of computation.

\subsection*{Key technical results}\label{sec:guide}

Here is a brief summary of the key technical results of this  paper. 

Mesonic gauge-invariant operators, composite fields made of $2n$ scalars $ (\Phi_{a})^i_j$, are associated to 
permutations $ \alpha \in S_{ 2n }$. They are written as $ \cO_{ \alpha }$. Different permutations $ \alpha $ related through conjugation by a permutation $ \gamma \in S_n [ S_2 ] $ are identical. 
\bea 
\cO_{ \alpha } = \cO_{ \gamma \alpha \gamma^{-1}}   ~~ \hbox{for } ~~ \gamma \in S_n [ S_2 ] 
\eea
We denote the two-point functions of normal-ordered operators $\cO_\alpha$ of length $2n$ by
\begin{equation}
\vev{\cO_{\alpha_1}(x) \cO_{\alpha_2} (y)} = \frac{\vev{\cO_{\alpha_1} \cO_{\alpha_2}}}{\abs{x-y}^{4n}} \,,
\end{equation}
where the bracket $\vev{ \cO_{\alpha_1} \cO_{\alpha_2} }$ contains the sum of all Wick contractions between the two operators.
The two-point functions $\vev{ \cO_{\alpha_1} \cO_{\alpha_2} }$ in the free theory will be described by an elegant formula expressed in terms of cycle structures of permutations (\ref{2pt formula2}). However, $\vev{ \cO_{\alpha_1} \cO_{\alpha_2} }$ viewed as a function of the permutation equivalence classes, is not diagonal.

To find a basis of operators with diagonal two-point functions in the free theory, we use Fourier transformation on finite groups. Ordinary Fourier transformation can be written as 
\begin{equation}
f^K= \int d \theta  D^K (\theta) \, f(\theta), \qquad
D^K (\theta) = e^{iK\theta} 
\end{equation}
where $ D^K ( \theta )$ is the action of the $U(1)$ group element $e^{ i \theta }$ in the representation of 
charge $K$. 
Fourier transformation on finite groups is given by 
\begin{equation}
f^R_{ij}  = \sum_{g \in G} D_{ij}^R(g) f (g), \qquad
f(g) \in \bb{C}
\label{def:FTFG}
\end{equation}
where $D_{ij}^R(g)$ is the matrix element of $g$ in the representation $R$ and $i, j $ run over some basis 
of the representation, which we will choose to be an orthogonal basis. A similar transformation in the group algebra $\mC ( G )$, formed by linear combinations of group elements with complex coefficients, is 
\begin{equation}
 Q^R_{ij}  = \sum_{g \in G} D_{ij}^R(g) ~ g
\label{def:FTFG1}
\end{equation}
Specializing $G$ to the symmetric group, the irreducible representations $R$ will correspond to 
Young diagrams.

In this paper,  we find a new diagonal basis for the free-field two-point functions, labelled by two irreducible representations $R, \Lambda_1$ of the permutation group $S_{2n}$ and a multiplicity label $\tau$,
\begin{multline}
\cO^{R,\Lambda_1,\tau} = tr_{2n}(P^{R,\Lambda_1,\tau} 
~ \Phi_{a_1} \otimes \Phi_{a_1} \otimes \Phi_{ a_2} \otimes \Phi_{a_2} \otimes \cdots \otimes  \Phi_{a_n} \otimes \Phi_{ a_n})
\\
(\tau = 1,2, \dots, C(R,R,\Lambda_1))
\label{def:O rep basis}
\end{multline}
as in \eqref{def:OR,Lambda1,tau}. 
Here $R$ corresponds to a Young diagrams with $2n$ boxes, $\Lambda_1$ to an even Young diagram with $2n$ boxes, i.e. all the row lengths are even numbers, and $C(R,R,\Lambda_1)$ is the number of times $\Lambda_1$ appears in the irreducible decomposition of $ R \otimes R $. 
$P^{R,\Lambda_1,\tau}$ is a linear combination of the sums over the equivalence classes of permutations
\begin{align}
[\alpha]=\frac{1}{|S_{n}[S_2]|}\sum_{\gamma\in S_{n}[S_2]}
\gamma \alpha \gamma^{-1} , \qquad
|S_{n}[S_2]| = 2^n n!
\label{def:square bracket}
\end{align}
which live in the group algebra $\bb{C}[S_{2n}]$.
The equation \eqref{def:OR,Lambda1,tau} for $P^{R,\Lambda_1,\tau} $ is a generalisation of  \eqref{def:FTFG1}, and the equation \eqref{def:O rep basis} for $\cO^{R,\Lambda_1,\tau} $ is a generalisation of \eqref{def:FTFG}, where $f(g)$ has been replaced by gauge-invariant polynomials in $(\Phi_a)^i_j$ parametrised by permutations.
The coefficients in $P^{R,\Lambda_1,\tau}$ do not depend explicitly on $N_c$ or $N_f$\,, though $R, \Lambda_1$ are required to have a bound on the number of rows when  $N_f, N_c < 2n$.

The number of diagonal operators \eqref{def:O rep basis} is given by
\begin{align}
\hbox{ Number of mesonic operators } =
\sum_{
\substack{ R \\c_1(R)\le N_c} }
\sum_{
\substack{ \Lambda_1:even \\ c_1(\Lambda_1)\le N_f} }
C(R,R,\Lambda_1).
\label{number-of-ops-intro}
\end{align} 
The two-point functions of \eqref{def:O rep basis} are given by \eqref{two-pt_rep_basis},
\begin{align}
\left\langle \cO^{R,\Lambda_1,\tau} \cO^{R^{\prime},\Lambda_1^{\prime},\tau^{\prime}} \right\rangle
= \delta^{RR^{\prime}} \, \delta^{\Lambda_1\Lambda_1^{\prime}} \, \delta^{\tau\tau^{\prime}} \,
\left(\frac{(2n)!}{d_R}\right)^2 Dim(R) \, N_f^{n} \,
\omega_{\Lambda_1/2} (\Omega_{2n}^{(f)}) 
\label{diag 2pt intro}
\end{align}
and they vanish unless the corresponding representation labels on the two operators are identical. 
The normalisation factor $\omega_{\Lambda_1/2} (\Omega_{2n}^{(f)}) $ is a polynomial in $N_f$.

The wreath product group $S_n[S_2]$ in \eqref{def:square bracket} also appears as a symmetry of the Kronecker delta's used for the contraction of flavour indices.
The groups $S_{2n}$ and $S_n[S_2]$ form what is called a Gelfand pair $(S_{2n}, S_n[S_2])$.
A notable property of the Gelfand pair is that the reduction of an irreducible representation of $S_{2n}$ 
into the singlet representation of $S_n[S_2]$ is multiplicity-free.
This property plays an essential role in our construction of the diagonal operators (\ref{def:O rep basis}), which is in the trivial representation of $ S_n[S_2]$.

The free-field two-point functions and the number of operators \eqref{number-of-ops-intro} are closely related, which become manifest in the large $N_c, N_f$ limit. Indeed, this agreement opens up a novel interpretation as the partition function of a topological field theory with a discrete gauge group.

The representation basis highly constrains the one-loop mixing, shown in \eqref{rep-mix-sel-rule}. 
We then establish that the elements of the mixing matrix are non-zero for operators having a pair of representation labels 
$ R , R' $ related by the move of at most one box. This is a familiar fact from previous studies of one-loop mixing in representation bases.

\section{Singlet operators}\label{sec: operators}

\subsection{Mesonic operators and wreath-product permutation group }

We denote a hermitian scalar field of gauge theory by $\Phi = (\Phi_a)^i_j$ with $a=1,2, \dots , N_f$ and $i,j = 1,2, \dots N_c$\,. The $\cN=4$ SYM corresponds to $N_f=6$. 
The upper gauge indices are identified with the lower gauge indices, up to an ordering parametrised by a permutation $ \alpha $
\bea 
(\Phi_{a_1})^{i_1}_{ i_{ \alpha(1)}  } ( \Phi_{a_2} )^{ i_2}_{ i_{ \alpha(2)} } \cdots  
( \Phi_{ a_{2n}} )^{ i_{2n}}_{ i_{ \alpha(2n)}}
\eea
Regarding $( \Phi_a )^i_j$ as matrix elements of operators $ \Phi_a : V_N \rightarrow V_N$ and defining 
\bea 
\Phi_{ \vec a } = \Phi_{ a_1} \otimes \Phi_{ a_2}\otimes  \cdots \otimes  \Phi_{ a_{2n} } 
\eea
which are operators in $ V_N^{ \otimes 2n }$, 
we have 
\bea 
(\Phi_{a_1})^{i_1}_{ i_{ \alpha(1)}  } ( \Phi_{a_2} )^{ i_2}_{ i_{ \alpha(2)} } \cdots  
( \Phi_{ a_{2n}} )^{ i_{2n}}_{ i_{ \alpha(2n)}} 
= tr_{ V_N^{ \otimes 2n } } ( \alpha  ~ \Phi_{\vec a }   )  = tr_{2n} ( \alpha ~  \Phi_{ \vec a } ) 
\label{abbr trPhi}
\eea
with $ \alpha $ acting in the standard way on the  tensor product $V_N^{\otimes 2n }$. 
In the second equality, the trace has been  abbreviated as $ tr_{2n}$.

Mesonic operators are defined as gauge-invariant operators whose flavour indices are pairwise contracted.
The general mesonic operator can be written by a permutation $\alpha \in S_{2n}$,
\begin{align}
\cO_\alpha &= \Big( \prod_{k=1}^{n} \delta^{a_{2k-1} a_{2k}} \Big) \, tr_{2n} ( \alpha \, \Phi_{\vec a})
\notag \\
&= tr_{2n}  \( \alpha ~ \Phi_{a_1} \otimes \Phi_{a_1} \otimes \Phi_{ a_2} \otimes \Phi_{a_2} \otimes \cdots \otimes  \Phi_{a_n} \otimes \Phi_{ a_n} \).
\label{def:mesonic}
\end{align}
Generally this is a multi-trace operator, whose trace structure is given by the cycle type of $\alpha$.
If the cycle type of $ \alpha $ is $p$, 
\begin{equation}
p = [1^{p_1}, 2^{p_2}, \dots, (2n)^{p_{2n}} ], \qquad \sum_{i=1}^{2n} i \, p_i = 2n \,,
\label{def:p cycle-type}
\end{equation}
then the number of traces of $\cO_\alpha$ is equal to the number of cycles in $\alpha$,
\begin{equation}
C(\alpha) = \sum_i p_i \,.
\label{def:number of cycles}
\end{equation}

The flavour contractions appear in pairs $(2k-1,2k)$, which is invariant under the permutation $\Sigma_0 = ( 1,2) ( 3,4) \cdots ( 2n-1, 2n )$ consisting of pairwise swops.
The permutation $\Sigma_0$ is invariant when conjugated by another set of permutations $ \gamma $ belonging to a subgroup $ S_n [ S_2 ] $ of $ S_{2n}$ 
\bea 
\gamma \Sigma_0 \gamma^{-1} = \Sigma_0 ~~~~ \hbox{ for } ~~~~ \gamma \in S_n [S_2] .
\eea
This wreath product group $S_n [ S_2]$\footnote{$S_n [ S_2 ]$ is also called the hyperoctahedral group in mathematics.}  has order $2^n n!$
\bea 
| S_n[S_2] | = 2^n n! 
\eea
It contains each of the $n$ pairwise swops, which form a subgroup $(S_2)^{ \times n } \subset S_n [ S_2] $, along with $n!$ permutations of the $n$ pairs.
The operator $ \cO_{ \alpha } $ is invariant under conjugation by $S_n [ S_2 ]$,  
\bea
\cO_{ \alpha } = \cO_{ \gamma \alpha \gamma^{-1} } ~~~~ \hbox{ for } ~~~~\gamma \in S_n [ S_2] 
\label{wreath equiv class}
\eea
because conjugation by a permutation gives re-ordering of the flavour indices, \eqref{reordering}.
Recalling the definition of $[\alpha]$ in \eqref{def:square bracket}, we observe
\begin{align} 
[\alpha]=[\gamma \alpha \gamma^{-1}]=\gamma [\alpha] \gamma^{-1} ~~~~ \hbox{ for } ~~~~ \gamma \in S_n [S_2] . 
\end{align}
We define 
\begin{equation}
\cO_{ [\alpha] } \equiv \frac{1}{|S_{n}[S_2]|}\sum_{\gamma\in S_{n}[S_2]}
\cO_{\gamma \alpha \gamma^{-1} }
\end{equation}
and observe that 
\begin{align}
\cO_{ \alpha } = \cO_{ [\alpha] } 
\label{defOalpha}
\end{align}
It follows therefore that gauge-invariant mesonic operators are in 1-1 correspondence with the sums $[\alpha]$ in the group algebra.

For $N_c \ge 2n$ and $N_f \ge 2n$, the mesonic operators of the form \eqref{defOalpha} is 
uniquely and completely specified by 
the equivalence classes \eqref{wreath equiv class}. 
The latter number is given by the Burnside Lemma as
\bea 
\hbox{ Number of mesonic operators } = { 1 \over \abs{S_n[S_2]} } \sum_{ \gamma \in S_n[S_2] }
 \sum_{ \alpha \in S_{2n} } \delta_{2n} ( \gamma \alpha \gamma^{-1} \alpha^{-1} ) ,
\eea
where
\begin{equation}
\delta_{2n} (g) = \begin{cases}
1 &\qquad (g= {\rm identity} \in S_{2n}) \\
0 & \qquad ({\rm otherwise}).
\end{cases} 
\label{def:delta 2n}
\end{equation}
Consider $\cO_\alpha$ with fixed trace structure, where the cycle type of $ \alpha $ is $p$ in \eqref{def:p cycle-type}. The number of such gauge-invariants is
\bea 
{ 1 \over \abs{S_n[S_2]} } \sum_{ \gamma \in S_n[S_2] }
 \sum_{ \alpha \in T_p } \delta_{2n} ( \gamma \alpha \gamma^{-1} \alpha^{-1} ) ,
\label{number of MO fixed p} 
\eea
where $T_p$ consists of permutations of cycle type $p$.
This is also equal to
\bea 
\frac{1}{|H_p|} \sum_{\sigma\in H_p}\varphi(\sigma), \qquad
\varphi ( \sigma ) = \sum_{ \gamma \in [2^n] } \delta_{2n} ( \gamma \sigma \gamma^{-1} \sigma^{-1}  ) ,
\label{number of MO fixed p2} 
\eea
where the elements of $H_p$ commute with a fixed permutation of cycle type $p$.
A derivation of the equality of these formulae, 
along with counting at finite $N_f $ is in section \ref{sec:finite Nf}.

When $N_c < 2n$ or $N_f < 2n$, there exist a number of 
linearly dependent relations among operators. 
They are called  finite $N_c$ constraints or finite $N_f$ constraints. 
These constraints can be expressed in terms of 
Young diagrams.\footnote{
Here is a simple example of a finite $N_c$ constraint. 
For a $2\times 2$ matrix $X$, 
we have the identity
\begin{align}
tr(X^3)=\frac{3}{2} \, trX \,tr(X^2)-\frac{1}{2}(trX)^3. 
\end{align}
We can rewrite this identity in terms of the projection operator 
associated with the anti-symmetric representation as 
\begin{align}
tr_3 (p_{[1^3]} X^{\otimes 3})=0. 
\end{align}
In short, we cannot anti-symmetrise more than $N_c=2$ indices.
}  
In section \ref{sec: rep basis}, we will construct a set of operators 
with representation labels, where the finite $N_c$ and finite $N_f$ constraints are manifest.


\subsection{Two-point functions}

Consider the free two-point functions of the mesonic operators \eqref{def:mesonic}. Using the Wick contraction rule \eqref{UNc Wick}, we obtain
\begin{align}
\contraction[3ex]{\langle \cO_{\alpha_1} \cO_{\alpha_2} \rangle = \delta_{\vec a} \, \delta^{\vec b} \ \Big\langle \prod_{k=1}^{2n}  }{(\Phi_{a_k})}{^{i_k}_{i_{\alpha_1 (k)}} \prod_{m=1}^{2n} }{(\Phi_{b_m})}
\langle \cO_{\alpha_1} \cO_{\alpha_2} \rangle 
&= \delta_{\vec a} \, \delta^{\vec b} \ 
\Big\langle \prod_{k=1}^{2n}  (\Phi^{a_k})^{i_k}_{i_{\alpha_1(k)}} 
\prod_{m=1}^{2n} (\Phi_{b_m})^{j_m}_{j_{\alpha_2 (m)}} \Big\rangle
= \delta_{\vec a} \, \delta^{\vec b} \sum_{\sigma \in \cS_{2n} } \prod_{k=1}^{2n} \delta^{a_k}_{b_m} \delta^{i_k}_{j_{\alpha_2 (m)}} \delta_{i_{\alpha_1 (k)}}^{j_m} \Big|_{m = \sigma^{-1} (k)} 
\notag\\
&= \delta_{\vec a} \, \delta^{\vec b} \sum_{\sigma \in \cS_{2n} }
\prod_{k=1}^{2n} \delta^{a_k}_{b_{\sigma^{-1} (k)}} \delta^{i_k}_{i_{\alpha_1 \sigma \alpha_2 \sigma^{-1} (k)}} 
\label{mtr Wick formula}
\end{align}
where $\sigma$ represents all possible Wick contractions, and 
$\delta_{\vec a} = \prod_{i=1}^n \delta_{a_{2i-1} a_{2i}}$.
The results can be expressed by permutations,\footnote{$\tr_{2n} (\sigma) = \tr_{2n} (\sigma \, 1)$ is a special case of \eqref{abbr trPhi}}
\begin{equation}
\langle \cO_{\alpha_1}\cO_{\alpha_2}\rangle 
= \sum_{\sigma\in S_{2n}} W(\sigma^{-1}) \, tr_{2n} (\alpha_1 \sigma \alpha_2 \sigma^{-1})
= \sum_{\sigma\in S_{2n}} W(\sigma^{-1}) \, N_c^{C(\alpha_1 \sigma \alpha_2 \sigma^{-1})} \,.
\label{two-pt-formula}
\end{equation}
Here $C(\sigma)$ is \eqref{def:number of cycles}, and $W(\sigma)$ is the flavour factor
\begin{align}
W(\sigma) = \delta_{\vec a} \, \delta^{\vec b} \ \( \prod_{k=1}^{2n} \delta^{a_k}_{b_{\sigma(k)}} \)
\equiv \delta_{\vec a} \, \delta^{\vec b} \ ( \sigma )^{ a_1, a_2 , \cdots , a_{2n}}_{ b_1 , b_2 , \cdots , b_{ 2n} } \,.
\label{def:Wsigma}
\end{align} 
We have $W(1)=N_f^{n}$. 
The $\sigma$ in \eqref{def:Wsigma} acts on $ V_{F}^{ \otimes 2n } $, with matrix elements equal to Kronecker deltas.
We can also write $ W( \sigma)$ by introducing contraction operators 
$ C_{12}, C_{34} \cdots $, where $ C_{12}$ acts on the first and second 
tensor factors of $V_{F}^{ \otimes 2n }$ as 
\bea 
C_{12} \, e_{a_1} \otimes e_{a_2}  = \delta_{ a_1 a_2} \, e_b \otimes e_b
\label{def:contraction operator}
\eea
Then 
\bea 
W ( \sigma )  = tr_{ V_F^{ \otimes 2n}  } 
\left ( C_{12}C_{34}C_{56} \cdots C_{2n-1,2n} \sigma \right ).
\eea
The contractions form part of the Brauer algebra which is the commutant of $O(N_f)$ acting on $V_F^{ \otimes 2n  }$, the $2n$-fold tensor product of the fundamental representations.

\begin{figure}[t]
\begin{center}
\includegraphics[scale=0.7]{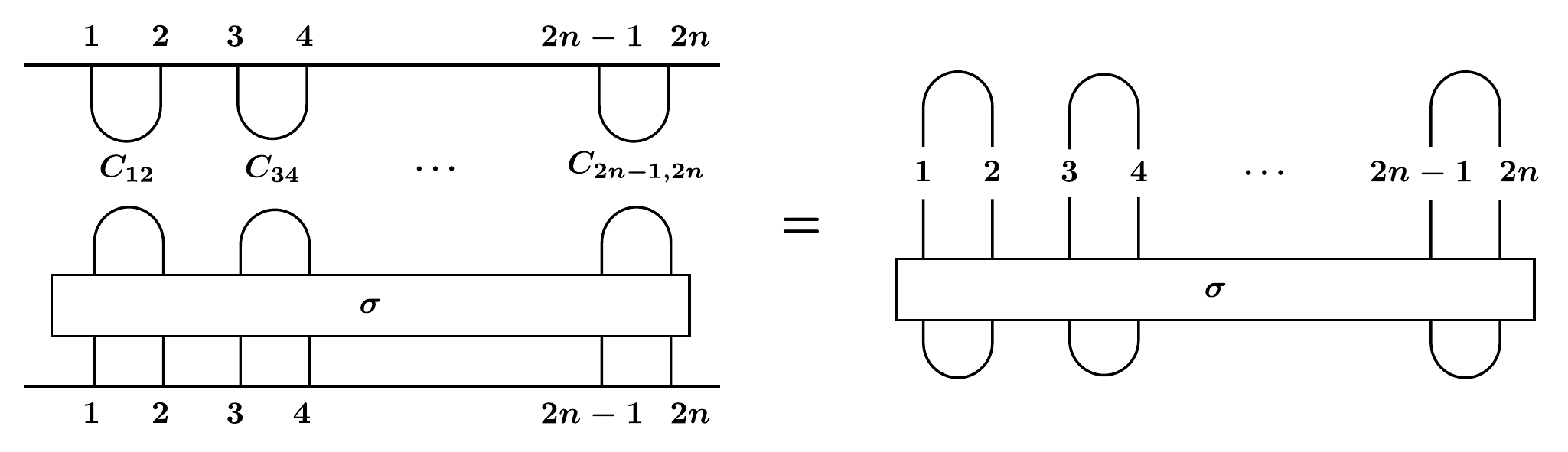}
\caption{The left figure shows $C_{12}C_{34} \cdots C_{2n-1,2n} \sigma$ acting on $V_F^{ \otimes 2n }$. The upper and lower horizontal lines are identified when taking the trace, as in the right figure.}
\label{fig:W(sigma)}
\end{center}
\end{figure}

\begin{figure}[t]
\begin{center}
\includegraphics[scale=0.75]{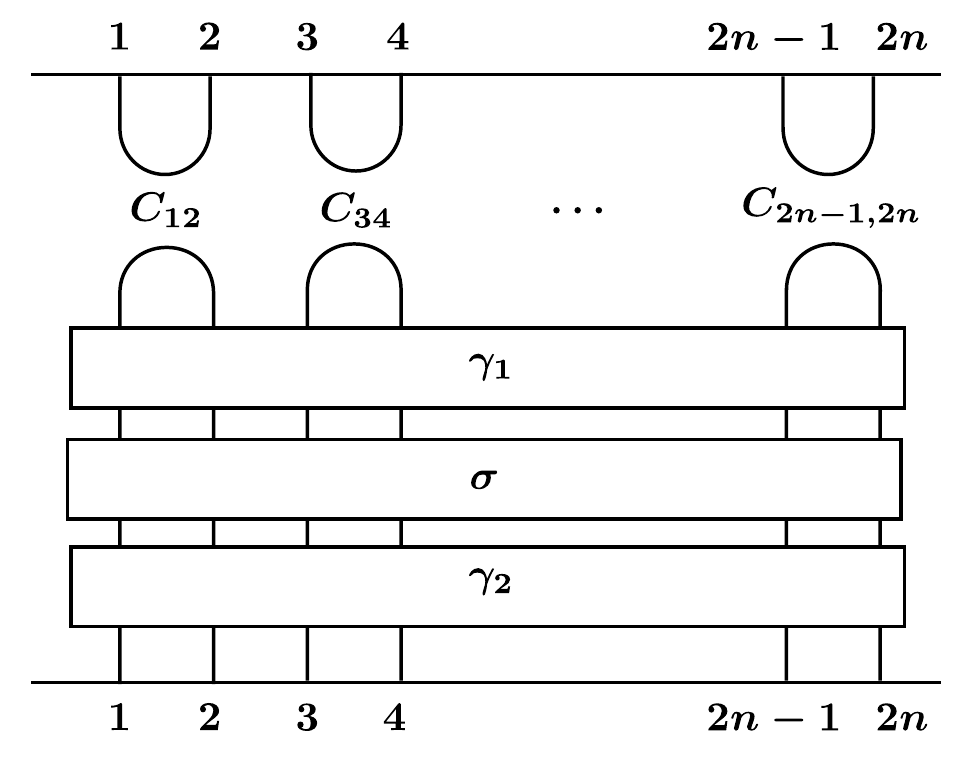}
\hspace{3mm}
\includegraphics[scale=0.75]{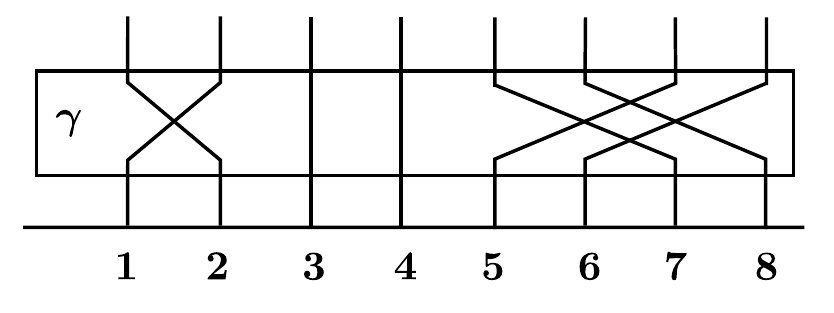}
\caption{(Left) Diagram for $W(\gamma_1 \sigma \gamma_2)$, which counts the number of connected components. (Right) An example of $\gamma \in S_n [S_2]$, showing that $S_n[S_2]$ does not change the number of connected components.}
\label{fig:Sn[S2]}
\end{center}
\end{figure}

The contraction operators obey 
\bea 
\gamma_1 ( C_{12} C_{34} \cdots C_{ 2n-1 , 2n } ) \gamma_2
=  ( C_{12} C_{34} \cdots C_{ 2n-1 , 2n } )
\eea
as can be seen from Figure \ref{fig:Sn[S2]}. It follows that 
\begin{align}
W(\sigma)=W(\gamma_1 \sigma \gamma_2), \qquad \forall \gamma_1,\gamma_2 \in S_n [S_2]
\label{equivalence_classes_W}
\end{align}
and from \eqref{def:Wsigma}
\begin{align}
W(\sigma)=W(\sigma^{-1}) .
\end{align}

Let us denote by $z(\sigma)$ the number of cycles $W(\sigma)$ represented in Figure \ref{fig:W(sigma)}.
Then $W(\sigma)$ is given by
\begin{equation}
W(\sigma)=N_f^{z(\sigma)}. 
\label{W_sigma_brauer}
\end{equation}
Indeed, the function $z(\sigma)$ satisfies
\begin{align}
z(\sigma)=z(\gamma_1 \sigma\gamma_2), \qquad \forall \gamma_1,\gamma_2 \in S_{n}[S_2],
\label{symmetry_z(sigma)}
\end{align}
We will show in Appendix \ref{sec:powers Nf} that it can be expressed as 
\begin{align}
z(\sigma)=\frac{1}{2} \, C (\Sigma_0\sigma\Sigma_0\sigma^{-1}).
\label{W_sigma_symmetric_group}
\end{align}
Since permutations can be multiplied efficiently using group theory software such as GAP or with Mathematica, 
this is a very useful expression for practical calculations. It does not involve explicitly doing sums 
over indices ranging from $1$ to $N_f$. 
The quantity (\ref{W_sigma_symmetric_group})  also has a nice mathematical meaning. 
It is equal to the number of cycles in the coset type of $\sigma$ as explained in Appendix \ref{app:Gelfand}.

\bigskip
For later purposes, let us define
\begin{align}
\Omega_{2 n}^{(f)}=\frac{1}{N_f^{n}}\sum_{\sigma \in S_{2n}}N_f^{z(\sigma)} \sigma^{-1} ,\qquad
\Omega_{2 n} =\frac{1}{N_c^{2n}}\sum_{\sigma \in S_{2n}}N_c^{C(\sigma)} \sigma^{-1} .
\label{def:Omega factors}
\end{align}
In terms of these quantities 
the two-point functions (\ref{two-pt-formula}) have the form of
\begin{align}
\langle \cO_{\alpha_1}\cO_{\alpha_2}\rangle 
&= \sum_{\sigma\in S_{2n} } N_f^{z (\sigma)} N_c^{C (\alpha_1 \sigma \alpha_2 \sigma^{-1})}
 ~ = ~ \sum_{\sigma\in S_{2n} } N_f^{{ 1 \over 2 } C ( \sigma \Sigma_0 \sigma^{-1} \Sigma_0 )  } N_c^{C (\alpha_1 \sigma \alpha_2 \sigma^{-1})} \notag \\
&= \sum_{\sigma}N_c^{2 n}N_f^{n} \, \delta_{2n} (\Omega_{2n}^{(f)}\sigma) \, 
\delta_{2n} (\Omega_{2n} \alpha_1 \sigma \alpha_2 \sigma^{-1}). 
\label{2pt formula2}
\end{align}
From (\ref{symmetry_z(sigma)}), the $\Omega_{2 n}^{(f)}$ satisfies
\begin{align}
\gamma_1 \Omega_{2n}^{(f)}\gamma_2=\Omega_{2n}^{(f)},
\end{align}
and we recover $W(\sigma)$ by
\begin{equation}
W(\sigma)=N_f^{n} \, \delta_{2n}(\Omega_{2 n}^{(f)}\sigma).
\label{W_sigma_Omega_delta}
\end{equation}


\section{Orthogonal two-point functions}\label{sec: rep basis}

We now propose a set of operators labelled by representations,
\begin{equation}
\begin{aligned}
\cO^{R,\Lambda_1,\tau} &= tr_{2n}(P^{R,\Lambda_1,\tau} ~ \Phi_{a_1} \otimes \Phi_{a_1} \otimes \Phi_{ a_2} \otimes \Phi_{a_2} \otimes \cdots \otimes  \Phi_{a_n} \otimes \Phi_{ a_n} ),
\\
P^{R,\Lambda_1,\tau} &= \sum_{\alpha\in S_{2n}}
B^{\Lambda_1}_k ~ 
\Clebsch{\tau}{\Lambda_1}{R}{R}{k}{i}{j} \, 
D^R_{ij} (\alpha^{-1}) ~ [\alpha],
\end{aligned}
\label{def:OR,Lambda1,tau}
\end{equation}
where $R$, $\Lambda_1$ are irreducible representations of $S_{2n}$, $\tau$ runs over $1, \dots, C(R,R,\Lambda_1)$ is the Clebsch-Gordan multiplicity for $ S_{2n}$ tensor products (see \eqref{def:CG number}), and the $\Clebsch{\tau}{\Lambda_1}{R}{R}{k}{i}{j}$ is the Clebsch-Gordan coefficient 
(\ref{def:CG coeff}). The use of this  group theory data in the construction of covariant bases was developed in \cite{BHR07,BHR08}. 
The $B^{\Lambda_1}_k$, called a branching coefficient, is defined in terms of 
the reduction of the irreducible  representation $ \Lambda_1$ of $ S_{2n}$ in terms of 
a direct sum of irreducible representations of $ S_n [ S_2] \subset S_{2n} $. A generic state $ | \Lambda_1 , k \rangle $ can be expanded in terms of irreducible representations of $S_n [ S_2]$. 
The one-dimensional irreducible representation of $ S_n [ S_2]$ is known  to appear in the decomposition of $ \Lambda_1$ with unit multiplicity if $ \Lambda_1$ is even, i.e. 
the partition of $\Lambda$ has the form 
$(2\lambda_1,2\lambda_2,\cdots)$ with integers $\lambda_i$.
Denoting this one-dimensional subspace of $ \Lambda_1 $ as  $| \Lambda_1 \rightarrow 1_{S_{n}[S_2]} \rangle $,  we have the branching coefficients
\begin{align}
B^{\Lambda_1}_k=\langle \Lambda_1,k \mid 
\Lambda_1 \rightarrow 1_{S_{n}[S_2]} \rangle,
\label{def:branching B}
\end{align}
where $k$ runs over $1,\cdots ,d_{\Lambda_1}$.\footnote{Due to Frobenius duality between restriction and induction, 
 the state $| \Lambda_1 \rightarrow 1_{S_{n}[S_2]} \rangle $ can also be thought as the  $\Lambda_1$-component of the induced representation ${\rm Ind}^G_H 1_H$ with $G=S_{2n}$ and $H=S_n[S_2]$.}
The branching coefficient is related to the singlet projector $p_{1_{S_n[S_2]}}$ by
\begin{equation}
D_{k_1k_2}^{\Lambda_1} (p_{1_{S_n[S_2]}}) 
= B^{\Lambda_1}_{k_1} B^{\Lambda_1}_{k_2} \,, \qquad
p_{1_{S_n[S_2]}} = \frac{1}{|S_n [S_2]|} \, \sum_{\gamma \in S_n[S_2]} \gamma \,.
\label{branching via projector}
\end{equation}

The quantity $P^{R,\Lambda_1,\tau}$ has the symmetry 
\begin{align}
P^{R,\Lambda_1,\tau} 
=\gamma P^{R,\Lambda_1,\tau}\gamma^{-1} 
\end{align}
for $\gamma \in S_n[S_2]$. From the inverse of \eqref{def:OR,Lambda1,tau}, $[\alpha]$ is expressed by 
\begin{align}
[\alpha]
=\frac{1}{(2n)!}\sum_{R,\Lambda_1,\tau }d_R 
B^{\Lambda_1}_k
\Clebsch{\tau}{\Lambda_1}{R}{R}{k}{i}{j} \,
D^R_{ji}(\alpha)P^{R,\Lambda_1,\tau} . 
\label{def:inverse transf}
\end{align}

Let us introduce 
\begin{align}
\omega_{\Lambda/2}(\sigma)=\frac{1}{|S_{n}[S_2]|}
\sum_{\gamma\in S_{n}[S_2]}
\chi_{\Lambda}(\sigma \gamma),
\label{zonal_S_n_wreath}
\end{align}
where $\Lambda$ is an even Young diagram \eqref{def:even rep}.
The function $\omega_{\Lambda/2}(\sigma)$ is called the zonal polynomial of 
the Gelfand pair $(S_{2n},S_{n}[S_2])$ \cite{MacdonaldBook}.
We have other expressions in terms of the projection operator 
associated with the singlet representation of $S_{n}[S_2]$,
\begin{align}
\omega_{\Lambda/2}(\sigma)
=
\chi_{\Lambda}(\sigma p_{1_{S_{n}[S_2]}})
=
\langle \Lambda \rightarrow 1_{S_{n}[S_2]}|
\sigma
|
\Lambda \rightarrow 1_{S_{n}[S_2]}\rangle. 
\label{zonal other}
\end{align}
The zonal polynomial has the property 
\begin{align}
\omega_{\Lambda/2}(\sigma)=\omega_{\Lambda/2}(\gamma_1 \sigma \gamma_2), 
\quad 
(\gamma_1,\gamma_2 \in S_{n}[S_2]).
\end{align}

The prominent property of the operators is that 
they have diagonal two-point functions  
\begin{align}
\langle \cO^{R,\Lambda_1,\tau}
O^{S,\Lambda_1^{\prime},\tau^{\prime}}
\rangle
=
\delta^{RS} \, \delta^{\Lambda_1\Lambda_1^{\prime}} \, \delta^{\tau\tau^{\prime}} \,
\left(\frac{(2n)!}{d_R}\right)^2 Dim(R) N_f^{n} \,
\omega_{\Lambda_1/2}(\Omega_{2n}^{(f)}),
\label{two-pt_rep_basis}
\end{align}
where $Dim(R)$ is the dimension of $U(N_c)$ associated with 
the representation $R$, and 
$d_R$ is the dimension of $S_{2n}$ associated with 
the representation $R$.
$\omega_{\Lambda_1/2} (\Omega_{2n}^{(f)})$ is a specialisation of the zonal spherical function of the Gelfand pair $(GL(N_f), O(N_f))$ explained in Appendix \ref{app:zonal}. It is a polynomial in $N_f$, which was also encountered in \cite{CDD1301,CDD1303}. This equation is derived in 
Appendix \ref{diagonal_twopoint_functions1}. 
Concrete examples of the operators $\cO^{R,\Lambda_1,\tau}$ are given in Appendix \ref{app:example}.

In section \ref{sec:counting_SW}
we count the number of mesonic operators using Schur-Weyl duality. 
It will turn out that the representation-labelled operator \eqref{def:OR,Lambda1,tau} should vanish if
\begin{equation}
c_1(R) < N_c \qquad {\rm or } \qquad c_1(\Lambda_1) < N_f \,,
\label{finite_Nc_finite_Nf}
\end{equation}
where $c_1(R)$ denotes the length of the first column of the Young diagram $R$.
This expectation is consistent with the normalisation of our two-point functions \eqref{two-pt_rep_basis}.
We can find the vanishing properties from the formulae 
\begin{equation}
Dim(R) = \prod_{(i,j) \in R} \frac{(N_c +j-i)}{h(i,j)}
\label{DimR formula}
\end{equation}
and 
\begin{equation}
\omega_{\Lambda_1/2} (\Omega_{2n}^{(f)}) 
= \frac{\abs{S_n[S_2]}}{N_f^n} \prod_{(i,j) \in \Lambda_1/2} (N_f + 2j-i-1),
\label{formula_character_omega_flavour}
\end{equation}
where $(i,j)$ labels the row and column of the boxes of the 
Young diagram, and $h(i,j)$ is the hook-length \eqref{def:hook-length}.
The formula (\ref{formula_character_omega_flavour}) was derived in \cite{MacdonaldBook,CDD1301}.
Both \eqref{DimR formula} and \eqref{formula_character_omega_flavour} depend on $N_c$ and $N_f$ through the product of factors, e.g.

\ytableausetup{mathmode, boxsize=2.5em, aligntableaux=center}
\begin{align}
Dim(R) &\sim \prod_{\rm boxes} \ 
\begin{ytableau}
\mbox{\scriptsize $N_c$} & \mbox{\scriptsize \ $N_c+1$ } & \mbox{\scriptsize \ $N_c+2$ } & \mbox{\scriptsize \ $N_c+3$ } \\
\mbox{\scriptsize \ $N_c-1$ } & \mbox{\scriptsize $N_c$ }
\end{ytableau} 
\\[2mm]
\omega_{\Lambda_1/2} (\Omega_{2n}^{(f)}) &\sim \prod_{\rm boxes} \ 
\begin{ytableau}
\mbox{\scriptsize $N_f$} & \mbox{\scriptsize \ $N_f+2$ } & \mbox{\scriptsize \ $N_f+4$ } & \mbox{\scriptsize \ $N_f+6$ } \\
\mbox{\scriptsize \ $N_f-1$ } & \mbox{\scriptsize $N_f+1$ }
\end{ytableau} 
\end{align}

\bigskip
In \cite{BHR08} a general formula of diagonal operators is presented, which works for any group $G$ and representation $V$. However, the formula requires the calculation of the CG coefficients arising from the decomposition of $ V_F^{ \otimes 2n}$ by $ G \times S_{2n}$\,.
It turns out that our operator $\cO^{R,\Lambda_1,\tau}$ is a concrete realisation of the general formulae of \cite{BHR08}. The relevant CG coefficients can be obtained explicitly 
as explained in Appendix \ref{sec:BHR}. The branching coefficients associated with the Gelfand pair $( S_{ 2n  } , S_{n}  [ S_2 ] )$ provide a  neat description of these Clebsch-Gordan coefficients, and hence of the space of $O(N_f)$ singlet operators of length $2n$.


\section{Operator counting}\label{sec: counting}

In this section we discuss several aspects of the counting formula, taking care of 
essential distinction between large $N_c$ and finite $N_c$ ( i.e. $N_c \ge 2n$ and $N_c < 2n$ ) 
as well as large $N_f$ and finite $N_f$ ( i.e. $N_f \ge 2n $ and $N_f < 2n$). 
First we employ Schur-Weyl duality and to obtain group theoretic expressions 
for  the dimension of the space of operators at finite $N_c, N_f$. Then we rewrite the counting formula in terms of delta functions of permutations in the large $N_c$ and  large $( N_c , N_f ) $ limits. The latter formulae are recognised as the counting formulae for graphs based on Burnside's Lemma. We will rediscover various ways of associating permutation labels to gauge-invariant operators at large $N_c, N_f$.

\subsection{Schur-Weyl duality and counting for finite $N_c , N_f$}
\label{sec:counting_SW}

The construction of our operators can be explained by Schur-Weyl duality, along the lines of \cite{BHR07,BHR08,Ramgoolam08,BrownThesis}. 
The scalar field $(\Phi^a)_i^j$ has one flavour and two colour indices, and belongs to $V_F \otimes V_C \otimes \ol{V}_C$\,. Thus the tensor product $\Phi^{\otimes 2n}$ belongs to $\( V_F \otimes V_C \otimes \ol{V}_C \)^{\otimes 2n}$. There is a natural action of permutations  $\sigma \in  S_{2n}$ on the tensor product,
\begin{equation}
\sigma \, : \, (\Phi^{a_1})_{i_1}^{j_1} (\Phi^{a_2})_{i_2}^{j_2} \dots (\Phi^{a_{2n}})_{i_{2n}}^{j_{2n}}
\ \mapsto \ 
(\Phi^{a_{\sigma(1)}})_{i_{\sigma(1)}}^{j_{\sigma(1)}} (\Phi^{a_{\sigma(2)}})_{i_{\sigma(2)}}^{j_{\sigma(2)}} \dots (\Phi^{a_{\sigma(2n)}})_{i_{\sigma(2n)}}^{j_{\sigma(2n)}} \,.
\label{S2n on tensor prod}
\end{equation}
The RHS is a polynomial of bosonic variables $(\Phi^a)_i^j$. Since they commute with each other, this polynomial is invariant under $\sigma$.

Let us count the number of gauge-invariant scalar operators. We regard $V_F, V_C, \ol{V}_C$ as $U(N_f), U(N_c) $-modules and apply Schur-Weyl (SW) duality. 
We find
\begin{multline}
\( V_F \otimes V_C \otimes \ol{V}_C \)^{\otimes 2n}
\\
= 
\mathop \bigoplus \limits_{\substack{\Lambda \vdash 2n \\ c_1 ( \Lambda_1 ) \le N_f}}
\mathop \bigoplus \limits_{\substack{R \vdash 2n \\ c_1 ( R )}}
\mathop \bigoplus \limits_{\substack{S  \vdash 2n \\ c_1 ( S ) \le N_c}}
\( V^{U (N_f)}_{\Lambda} \otimes V^{S_{2n}}_{\Lambda} \) \otimes
\( V^{U (N_c)}_{R} \otimes V^{S_{2n}}_{R} \) \otimes
\( \bar V^{U (N_c)}_{ S} \otimes  V^{S_{2n}}_{ S} \) .
\label{SWD tensor 2n}
\end{multline}
where $c_1(R)$ is defined in \eqref{def:c1Q}.
We impose the condition $R = S$ to select $U(N_c)$-invariant operators, 
\begin{equation}
V^{U(N_c)}_{R} \otimes \bar V^{U(N_c)}_{S} \Big|_{U(N_c)} = \delta_{R,  S} \, V^{U(N_c)}_{\emptyset}
\end{equation}
and the operators should also be $S_{2n}$-invariant
\begin{equation}
V^{S_{2n}}_{\Lambda} \otimes V^{S_{2n}}_{R} \otimes V^{S_{2n}}_{ R} \Big|_{S_{2n}}
= \bigoplus_{\Lambda'}     V^{S_{2n}}_{\Lambda} \otimes V_{ RR}^{\Lambda'} \otimes   V^{S_{2n}}_{\Lambda'} \Big|_{S_{2n}}
= V_{ RR}^{\Lambda} \otimes  V^{S_{2n}}_{\rep{2n}} \,.
\end{equation}
where $ V_{RR}^{ \Lambda'}$ is the multiplicity space of  $ V_{ \Lambda'}^{ S_{2n} } $ in the tensor product 
$ V_R^{S_{2n}}  \otimes V_R^{ S_{2n}} $.  Its dimension is the CG multiplicity $ C ( R , R , \Lambda' )$ 
\eqref{def:CG number}.
This implies 
\begin{equation}
\( V_F \otimes V_C \otimes \ol{V}_C \)^{\otimes 2n} \Big|_{U(N_c) \times S_{2n} }
= \mathop \bigoplus \limits_{\Lambda}
\mathop \bigoplus \limits_{R}
\( V^{U (N_f)}_{\Lambda} \otimes V^{\Lambda}_{RR}  \otimes V^{S_{2n}}_{\rep{2n}} \) .
\label{SWD general}
\end{equation}
The number of gauge-invariant scalar operators is given by taking the dimension of both sides.

Now we count the number of mesonic operators \eqref{def:mesonic}, which contains the Kronecker delta $\prod_{i=1}^n \delta^{a_{2i-1} a_{2i}}$. The mesonic operator is invariant under the following action of permutations $ \gamma$ in the  wreath product group $S_n[S_2]$,
\begin{multline}
 \gamma \, : \, \Bigl( \prod_{i=1}^n \delta^{a_{2i-1} a_{2i}} \Bigr) 
(\Phi^{a_1})_{i_1}^{j_1} (\Phi^{a_2})_{i_2}^{j_2} \dots (\Phi^{a_{2n}})_{i_{2n}}^{j_{2n}}
\\
\mapsto \ 
\Bigl( \prod_{i=1}^n \delta^{a_{\gamma(2i-1)} a_{\gamma(2i)}} \Bigr) 
(\Phi^{a_{\gamma(1)}})_{i_1}^{j_1} (\Phi^{a_{\gamma(2)}})_{i_2}^{j_2} \dots (\Phi^{a_{\gamma(2n)}})_{i_{2n}}^{j_{2n}} \,.
\label{SnS2 on tensor prod}
\end{multline}
This $\gamma$ is not part of $\sigma$ in \eqref{S2n on tensor prod}, because $\gamma$ does not change the colour indices.
Recall that the wreath product subgroup $ S_n [ S_2] $ is the set of permutations which leaves 
 $\Sigma_0 = (12)(34) \dots (2n-1,2n) \in S_{2n}$ invariant under action by conjugation
 \bea 
 \gamma \Sigma_0 \gamma^{-1} = \Sigma_0  ~~ \hbox{for} ~~  \gamma \in S_n [ S_2 ] .
 \eea
 The contraction of the flavour indices breaks $U(N_f)$ to $ O(N_f)$ and projects to the invariant representation of 
 $O(N_f)$.  So we can count the number of mesonic operators by restricting \eqref{SWD tensor 2n} to the  subspace invariant 
 under  $O(N_f) \times S_n[S_2] \times S_{2n} \times U( N_c) $. 
We use the fact that 
\begin{equation}
\begin{aligned}
V_F^{\otimes 2n} \big|_{O( N_f ) \times S_{n}[S_2]}
&= \bigoplus_{\Lambda} \ V_{\Lambda}^{U(N_f)}\big|_{O(N_f) }  \otimes  V_{\Lambda}^{S_{2n}} \big|_{S_{n}[S_2]} 
\\
&= \bigoplus_{\Lambda_1 : \, {\rm even}  } V_{\Lambda_1 }^{U (N_f)}\big|_{ O(N_f)}   \otimes  V_{\Lambda_1 }^{S_{2n}} \big|_{S_{n}[S_2]} .
\end{aligned}
\label{counting wreath-inv}
\end{equation}
Both $ (  GL( N_f) , O(N_f) ) $ and $ ( S_{2n } , S_n [ S_2 ] ) $ are Gelfand pairs \cite{MacdonaldBook}, which has the multiplicity-free property.
In particular $ \Lambda$ of the parent group contains the trivial of the subgroup 
with unit multiplicity if $ \Lambda_1$ is an even partition \eqref{def:even rep}.
Now we repeat the above argument, taking into account that the two factors in \eqref{counting wreath-inv} are one-dimensional,  to obtain
\begin{equation}
\( V_F \otimes V_C \otimes \ol{V}_C \)^{\otimes 2n} \Big|_{O(N_f) \times S_n [ S_2 ]  \times U(N_c) \times S_{2n}} = 
 \mathop \bigoplus \limits_{\Lambda_1 : \, { \rm even}  }
\mathop \bigoplus \limits_{R}  \, 
V_{RR}^{\Lambda_1 }   
\label{SWD rest}
\end{equation}
The number $I_{2n} (N_c, N_f)$ of mesonic singlet operators of length $2n$ at finite $N_c, N_f$ is thus 
\begin{equation}
I_{2n} (N_c, N_f) = \sum_{\substack{\Lambda_1 : \, {\rm even} \\ c_1 ( \Lambda_1 ) \le N_f}} 
\sum_{\substack{R \\ c_1 ( R ) \le N_c}} C(R, R, \Lambda_1) .
\label{SW counting}
\end{equation}
This result agrees with the number of diagonal operators \eqref{def:OR,Lambda1,tau}.
Table \ref{tab:I2n} shows some values of $I_{2n} (N_c, N_f)$ at large $N_c, N_f$.

The formula \eqref{SW counting} counts only the $O(N_f)$ singlets. 
The method to count the $SO(N_f)$ singlets or other representations are described in \cite{BHR08}. The difference between $O(N_f)$ and $SO(N_f)$ lies in the existence of ``baryonic'' operators. We will explain more about these points in Appendix \ref{sec:baryonic}.

\begin{table}[t]
\begin{center}
\begin{tabular}{c|cccccc}
$2n$ & 2 & 4 & 6 & 8 & 10 & 12 
\\\hline
$I_{2n}$ & 2 & 8 & 34 & 182 & 1300 & 12534
\end{tabular}
\caption{The number of all mesonic operators of length $2n$ for $N_c \ge 12$ and $N_f \ge 6$.}
\label{tab:I2n}
\end{center}
\end{table}


\subsection{Large $N_c$}\label{sec:finite Nc}

The number of singlet mesonic operators for finite $N_c$ and finite $N_f$ is given by \eqref{SW counting}.
We consider simplifications at large $N_c$; more precisely $ N_c \ge 2n $. 
There we can convert the sum over $R$ into a sum over permutations.

\bigskip
Let us define
\begin{align}
\varphi_{N_f}(\sigma):=
\sum_{
\substack{ \Lambda_1:even \\ c_1(\Lambda_1)\le N_f} }
\chi_{\Lambda_1}(\sigma) ,
\end{align}
and apply the formula (\ref{formula_two_characters_combine_delta}), valid for $N_c \ge 2 n$, to the counting formula \eqref{SW counting}.
It simplifies as
\begin{align}
I_{2n} (N_f) \equiv I_{2n} (N_c \ge 2n, N_f)
= \frac{1}{(2n)!}\sum_{\sigma\in S_{2n}}
\sum_{\gamma\in S_{2n}}\delta_{2n}
(\gamma \sigma \gamma^{-1} \sigma^{-1})\varphi_{N_f}(\sigma) 
\label{large_Nc_counting_finite_Nf} 
\end{align}
where $\delta_{2n} (g)$ is defined by \eqref{def:delta 2n}.
When the cycle type of $\sigma$ is $p=[1^{p_1},2^{p_2},\cdots,n^{p_n}]$ ($\sum_{i=1}ip_i=2n$), define
\begin{align}
T_p &= \{ \sigma \in S_{2n} \, | \, \text{cycle type of $\sigma$ is $p$} \}
\label{def:Tp}
\\[1mm]
H_p (\sigma) &= \{ \gamma \in S_{2n} \, | \, \gamma\sigma = \sigma \gamma , \ \sigma \in T_p \}.
\label{def:Hp}
\end{align}
Note that $H_p(\sigma), H_p(\sigma')$ are conjugate with each other if $\sigma, \sigma' \in T_p$. We define $H_p$ as $H_p (\sigma_*)$ for a fixed $\sigma_* \in T_p$\,. 
The order of $H_p$\,, namely the number of elements that commute with any permutation of cycle type $p$, is given by
\begin{align}
|H_p|=\frac{(2n)!}{|T_p|} \,.
\end{align}
See Table \ref{tab:pHp n=2} for examples. In particular, when the cycle type is $[2^n]$, the symmetry group (the stabiliser) is $H_{[2^n]} = S_n[S_2]$. Thus
\begin{align}
|S_n[S_2]|=\frac{(2n)!}{(2n-1)!!} = 2^n n!.
\label{wreath_product_2^n}
\end{align}

\begin{table}[t]
\begin{center}
\begin{tabular}{c|ccccc}
$p$ & $\rep{4}$ & $\rep{3,1}$ & $\rep{2^2}$ & $\rep{2,1^2}$ & $\rep{1^4}$  \\[1mm]
\hline
$H_p$ & $\bb{Z}_4$ & $\bb{Z}_3\times \bb{Z}_1$ & $S_2[\bb{Z}_2]$ & $\bb{Z}_2 \times S_2$ & $S_4$ 
\end{tabular}
\caption{The symmetry group which preserves the cycle type at $2n=4$.}
\label{tab:pHp n=2}
\end{center}
\end{table}

Then (\ref{large_Nc_counting_finite_Nf}) may be written as\footnote{Here we used the fact that any $\gamma\in T_p$ is written as $\gamma = \nu \gamma_* \nu^{-1}$ with a fixed $\gamma_* \in T_p$ for some $\nu \in S_{2n}$. This $\nu$ disappears after the redefinition $\sigma' = \nu^{-1} \sigma \nu$. The sum of $\varphi_{N_f}(\sigma)$ over $H_p (\gamma_*)$ does not depend on the choice of $\gamma_*$.\label{footnote:comp}}
\begin{align}
I_{2n} (N_f)
&=
\frac{1}{(2n)!}\sum_{\sigma\in S_{2n}}
\sum_{p \vdash 2n} \sum_{\gamma\in T_p}\delta_{2n}
(\gamma \sigma \gamma^{-1}\sigma^{-1})\varphi_{N_f}(\sigma)
\nonumber \\
&=
\frac{1}{(2n)!}\sum_p |T_p| \sum_{\sigma\in H_{p}}\varphi_{N_f}(\sigma)
\nonumber \\
&=
\sum_p \frac{1}{|H_p|} \sum_{\sigma\in H_{p}}\varphi_{N_f}(\sigma)
\label{finite_Nc_counting_1}
\end{align}
or 
\begin{align}
I_{2n} (N_f)
&=
\frac{1}{(2n)!}\sum_{q \vdash 2n} \sum_{\sigma\in T_{q}}
\sum_{\gamma\in S_{2n}}\delta_{2n}
(\gamma \sigma \gamma^{-1}\sigma^{-1})\varphi_{N_f}(\sigma)
\nonumber \\
&=
\frac{1}{(2n)!}\sum_q |H_q|\sum_{\sigma\in T_{q}}\varphi_{N_f}(\sigma)
\nonumber \\
&=
\sum_q \frac{1}{|T_q|} \sum_{\sigma\in T_{q}}\varphi_{N_f}(\sigma)
\label{finite_Nc_counting_2}
\end{align}
The two expressions give the large $N_c$ equivalence 
\begin{align}
I_{2n} (N_f)
=\sum_p \frac{1}{|H_p|} \sum_{\sigma\in H_{p}}\varphi_{N_f}(\sigma)
=
\sum_q \frac{1}{|T_q|} \sum_{\sigma\in T_{q}}\varphi_{N_f}(\sigma) .
\label{large Nc equivalence}
\end{align}
This equivalence can also be derived from 
the following group theory identity 
\begin{align}\label{pqtwosums} 
\sum_{p}\frac{1}{|H_p|}\sum_{\sigma\in H_{p}} 
\sum_{\mu\in S_{2n}}\mu \sigma \mu^{-1}
=\sum_q \frac{1}{|T_q|}
\sum_{\sigma \in T_q}
\sum_{\mu\in S_{2n}}\mu \sigma \mu^{-1}
\end{align}
The expressions \eqref{finite_Nc_counting_1} and \eqref{finite_Nc_counting_2} will be used in the next subsection.

\subsection{Large $N_c $ and $N_f$}\label{sec:finite Nf}

We take large $N_f$ and large $N_c$ limits to convert all representations to permutations, and obtain several expressions of the counting.

We define
\begin{equation}
\varphi(\sigma):=\varphi_{N_f \ge 2n}(\sigma) 
= \sum_{\Lambda_1:even} \chi_{\Lambda_1}(\sigma)
= \sum_{\Lambda_1 \vdash 2n} \chi_{\Lambda_1} (\sigma)M^{\Lambda_1}_{1_{S_{n}[S_2]}} .
\end{equation}
where we used the formula \eqref{mult-free SnS2} saying that $(S_{2n}, S_n[S_2])$ is a Gelfand pair. We also have
\begin{align}
\varphi(\sigma)
&=
\frac{1}{|S_n[S_2]|}\sum_{\Lambda_1}\sum_{u\in S_n[S_2]}
\chi_{\Lambda_1}(\sigma)
\chi_{\Lambda_1}(u)
\nonumber \\
&=
\frac{1}{|S_n[S_2]|}
\sum_{u\in S_n[S_2]}\sum_{\mu\in S_{2n}}\delta_{2n}(\mu \sigma \mu^{-1}u ) .
\label{derive_varphi_large_N_f_1}
\end{align}
Let us define
\begin{align}
Z_p \equiv \frac{1}{|H_p|}
\sum_{\sigma\in H_p}\varphi(\sigma),
\label{def:Zp varphi}
\end{align}
which is also written as
\begin{align}
Z_p =\sum_{\Lambda_1\vdash 2n}M^{\Lambda_1}_{1_{H_p}}
M^{\Lambda_1}_{1_{S_n[S_2]}}
=\sum_{\Lambda_1\vdash 2n,even}M^{\Lambda_1}_{1_{H_p}}
\end{align}
From (\ref{finite_Nc_counting_1}) the total number of singlets with $2n$ fields is
\begin{align}
I_{2n} \equiv I_{2n} (N_f \ge 2n)
=\sum_{p \vdash 2n} Z_p
\label{largeN_f_counting_ZHp}
\end{align}

We will derive three expressions for $Z_p$ below.\footnote{Each expression of $Z_p$ is related to different ways of writing the mesonic operators. Let us introduce
\begin{align}
\cO_{\alpha,\rho} =
\Bigl( \prod_{i=1}^n \delta^{a_{2i-1} a_{2i}} \Bigr) \, 
tr_{2n} \( \alpha \ \Phi_{a_{\rho(1)}}\otimes \Phi_{a_{\rho(2)}} \otimes  \cdots \otimes \Phi_{a_{\rho(2n-1)}}
\otimes \Phi_{a_{\rho(2n)}} \) .
\label{two_permutation_basis_alpha_rho}
\end{align}
This description is redundant in the following way,
\begin{equation}
\cO_{\alpha,\rho} = O_{\gamma_1 \alpha \gamma_1^{-1}, \gamma_2 \rho \gamma_1^{-1}} \qquad
\( \gamma_1 \in S_{2n} \,, \ \ \gamma_2 \in S_n[S_2] \).
\label{equivalence Oar}
\end{equation}
The permutation $\gamma_1$ comes from the re-ordering \eqref{reordering}, and $\gamma_2$ is the symmetry of the Kronecker delta's.
By using this redundancy we can gauge-fix either $\alpha$ or $\rho$. If we fix $\rho$, we obtain the operator $\cO_\alpha$ in \eqref{def:mesonic}. The corresponding counting formula is \eqref{counting_burnside_largeN_f}.
If we fix $\alpha$, then the cycle structure of $\alpha$, denoted by $p \vdash 2n$, determines the colour (or multi-trace) structure. It corresponds to the counting formula is \eqref{deltdc}. 
Finally, we use the contraction operators \eqref{def:contraction operator} and write
\begin{align}
\cO_{\alpha,\tau} =
\Bigl( \prod_{i=1}^n C_{\rho(2i-1) \rho(2i)} \Bigr) \, tr_{2n} \( \alpha \ \Phi_{a_1} \otimes \dots \otimes \Phi_{a_{2n}} \) 
\label{two_permutation_basis_alpha_tau}
\end{align}
We find that the transformation rule of the quantities $\prod_{i=1}^n C_{\rho(2i-1) \rho(2i)} = \rho^{-1} \( \prod_{i=1}^n C_{2i-1,2i} \) \rho$ and $\tau = \rho^{-1} \Sigma_0 \rho$ under the map $\rho \to \rho \gamma^{-1}$ are identical. The expression \eqref{two_permutation_basis_alpha_tau} is related to the last counting formula \eqref{counting_burnside_largeN_f-1}.
}

The quantity $Z_p$ is the number of equivalence classes in the double coset $S_{n}[S_2] \backslash S_{2n } / H_p$,
\begin{equation}
Z_p = \frac{1}{|S_n[S_2]|}\frac{1}{|H_p|}
\sum_{\sigma\in H_p}
\sum_{u\in S_n[S_2]}\sum_{\mu\in S_{2n}}\delta_{2n}(u \mu \sigma \mu^{-1} )
\label{deltdc} 
\end{equation}
The double coset space is the set of equivalence classes of permutations in $ S_{2n}$, generated by the left and right multiplications by the subgroups $ S_n [ S_2 ] $ and $H_p$ respectively. 
\bea 
\mu \sim u \mu \sigma ~~~~~ \( u \in S_n[S_2],   ~~ \sigma \in H_p \).
\eea
The above delta-function sum (\ref{deltdc}), counting the number of elements in the double coset space, 
is the application of the Burnside Lemma, which reduces the counting of orbits under a group action to 
the counting of fixed points under the group action. This double coset counting is the same as the counting 
of bi-partite graphs.
The bi-partite graphs have vertices in two colours (say black and white), and edges which connect only the vertices of different colours. The black vertices associated with partitions $p$, have cyclic order and there are $p_1$ univalent, $p_2$ bi-valent, $p_3 $ trivalent vertices, etc. These are easy to understand in terms of counting of traces of the scalar fields with global symmetry indices contracted. A cyclic black vertex of valency $v$ corresponds to a trace with $v$  scalar fields. The white vertices correspond to 
links between pairs of edges emanating from the black vertices, and correspond to flavour indices of the corresponding fields being contracted. 
The connection between double cosets and graph counting is explained in a physics context in \cite{Feynstring}. By going 
to large $N_c, N_f$, we see that counting $SO(N_f)$ invariants is simply counting the graphs.

Now observe that the last line 
in (\ref{derive_varphi_large_N_f_1})
can be rewritten as\footnote{
At the first line of (\ref{derive_varphi_large_N_f_2})
we have used that 
the elements in $S_{n}[S_2]$ satisfy 
\begin{align}
\Sigma_0 u \Sigma_0 ^{-1}=u
\end{align}
for $\Sigma_0=(12)(34)\cdots (2n-1,2n)$. 
See also the discussion in section 5.4 of \cite{Feynstring}.
} 
\begin{align}
\varphi(\sigma)=&
\frac{1}{|S_n[S_2]|}
\sum_{u\in S_{2n}}\sum_{\mu \in S_{2n}}
\delta_{2n} (\mu\sigma \mu^{-1}u)
\delta_{2n} (\Sigma_0 u \Sigma_0^{-1}u^{-1})
\nonumber \\
=&
\frac{1}{|S_n[S_2]|}
\sum_{\mu \in S_{2n}}
\delta_{2n} (\Sigma_0 
\mu\sigma \mu^{-1}
 \Sigma_0^{-1}
\mu\sigma^{-1} \mu^{-1}
)
\nonumber \\
=&
\sum_{\tau\in [2^{n}]}\delta_{2n} (\tau\sigma\tau^{-1}\sigma^{-1}).
\label{derive_varphi_large_N_f_2}
\end{align}
From \eqref{large_Nc_counting_finite_Nf}, \eqref{largeN_f_counting_ZHp} and \eqref{derive_varphi_large_N_f_2}, 
\begin{align}
Z_p &= \frac{1}{(2n)!} \sum_{ \sigma \in T_p } 
\sum_{\gamma\in S_{2n}}\delta_{2n}
(\gamma \sigma \gamma^{-1} \sigma^{-1})
\varphi( \gamma )
\nonumber\\
&= \frac{1}{(2n)!}
\sum_{\sigma\in T_p}
\sum_{\gamma\in S_{2n}}\delta_{2n}
(\gamma \sigma \gamma^{-1}\sigma^{-1})
\sum_{\tau \in [2^{n}]}\delta_{2n}(\tau\gamma\tau^{-1}\gamma^{-1}).
\label{counting_burnside_largeN_f-1}
\end{align}
This can be recognised as the counting of pairs $ ( \sigma , \tau ) $ in conjugacy classes $ ( T_p , [2^n] ) $, 
subject to equivalences $ ( \sigma , \tau ) \sim ( \gamma \sigma \gamma^{-1} , \gamma \tau \gamma^{-1} ) $ 
where $ \gamma \in S_{2n}$.  Such equivalence classes of pairs form another way of encoding bi-partite graphs. 
It amounts to choosing a
labelling of the edges using integers $ \{ 1, \cdots , 2n \}$ and reading off the labels of the edges around the black and white vertices. This is an alternative way to encode graphs, which differs from the encoding by a permutation $ \sigma \in S_{ 2n}$ which links directly with the 
counting by double cosets. This way of encoding graphs, in the context of Feynman graphs (which have symmetric rather than the cyclic vertices 
 here) is illustrated in Figure 7 of \cite{Feynstring}). The way that links directly with double cosets is shown in Figure 10 there.

Some further manipulation of \eqref{counting_burnside_largeN_f-1} gives
\begin{align}
Z_p &= \frac{1}{|S_n[S_2]|} \sum_{\sigma\in T_p} \sum_{\gamma\in S_n[S_2]}
\delta_{2n} (\gamma \sigma \gamma^{-1}\sigma^{-1} ).
\label{counting_burnside_largeN_f}
\end{align}
The equation \eqref{counting_burnside_largeN_f} establishes the equivalence between \eqref{number of MO fixed p} and \eqref{number of MO fixed p2}. We can reproduce this result also by applying the Burnside Lemma directly to the equivalence class of \eqref{equivalence Oar}.

Using (\ref{finite_Nc_counting_2}) we obtain yet another formula
\begin{align}
I_{2n} = \sum_q \tilde Z_q, \qquad
\tilde Z_q = \frac{1}{\abs{T_q}} \sum_{\gamma \in T_q} \sum_{\tau\in [2^{n}]}\delta_{2n} (\tau \gamma \tau^{-1} \gamma^{-1}).
\end{align}
Note however that $\tilde Z_q \neq Z_q$.\footnote{For example, when $q=\{ \rep{4}, \rep{3,1}, \rep{2^2}, \rep{2,1^2}, \rep{1^4} \}$, then $\tilde Z_q = \{1,0,3,1,3\}$ and $Z_q = \{2,1,2,2,1\}$.} 
We now have some physical insight into the two ways of writing $ I_{2n}$ as sums over partitions in (\ref{pqtwosums}). 
\eqref{large Nc equivalence}
In one way we have $ Z_p$. In another, we have $ \varphi ( \gamma )$ with $\gamma \in T_q$. The partition $p$ 
is the trace structure. The partition $q$ is the cycle structure of $\gamma$ which commutes with $ \tau \in [2^n ] $. 
Note that we have arrived at the sums over a product of two delta functions in this section 
by taking large $N_c, N_f $.

It is instructive to reconsider in reverse what we did in this section. 
Start from gauge-invariant operators parametrised by permutation equivalence classes.
Graph counting associated with gauge-invariant operators 
can be expressed in terms of permutation sums  with a product of delta functions. 
The Young diagrams $R , \Lambda $ come from Fourier transforming these two delta functions. 
In the present context, we have seen that the numbers of rows of $R , \Lambda$ are cut off by the rank of gauge and global symmetry groups.
In a wider context, we may wonder about the physical meaning of introducing extra integers to cut off the numbers of columns.


\section{Permutation topological field theory}\label{sec:TFT}

We explore the connection to a two-dimensional topological field theory (TFT) of permutation groups. 
This TFT is a topological lattice gauge theory whose gauge group is $S_{2n}$, defined on a 2-complex (collection of 0- , 1- and 2-cells glued together). 
The computation of observables of the TFT involves a sum over the group elements of $S_{2n}$ for every edge ($1$-cell), with a weight equal to a product of  delta functions, one for every face ($2$-cell). 
The delta function weight ensures that the sum 
is invariant under refinement of the cell decomposition, so that a continuum limit can be reached. This type of TFT is discussed in the physics literature in e.g. \cite{DW,FHK,CFS}. A review of the TFT of permutations and application to a large class of observables in quiver gauge theories is given in \cite{quivcalc}.

As a first step, we reconsider the number of mesonic states as a partition function in TFT. 
By summing over $p \vdash 2n$ in \eqref{counting_burnside_largeN_f-1}, we find
\begin{align}
I_{2n} 
&= \frac{1}{|S_n[S_2]|} \sum_{\sigma\in S_{2n}} \sum_{\gamma\in S_{2n}}
\delta_{2n} (\gamma \sigma \gamma^{-1}\sigma^{-1} ) \delta_{2n}(\Sigma_0 \gamma \Sigma_0 \gamma^{-1}) 
\label{TFT part fn1} 
\end{align}
This formula gives the number of mesonic singlets as a partition function for $S_{2n}$ TFT on the 2-complex
shown in the upper left of Figure \ref{fig:2PTFT2}. The 2-complex consists of two tori (drawn as cylinders with top 
and bottom boundaries identified) joined along  a circle, associated with permutation $ \sigma $. One of the tori has a cycle with permutation $ \Sigma_0$, which is fixed rather than being summed. This cycle with constrained permutation is a defect. 
Note that the delta function $\delta_{2n}(\Sigma_0 \gamma \Sigma_0 \gamma^{-1})$ can be solved explicitly as in \eqref{counting_burnside_largeN_f}.

\begin{figure}[t]
\begin{center}
\includegraphics[scale=0.7]{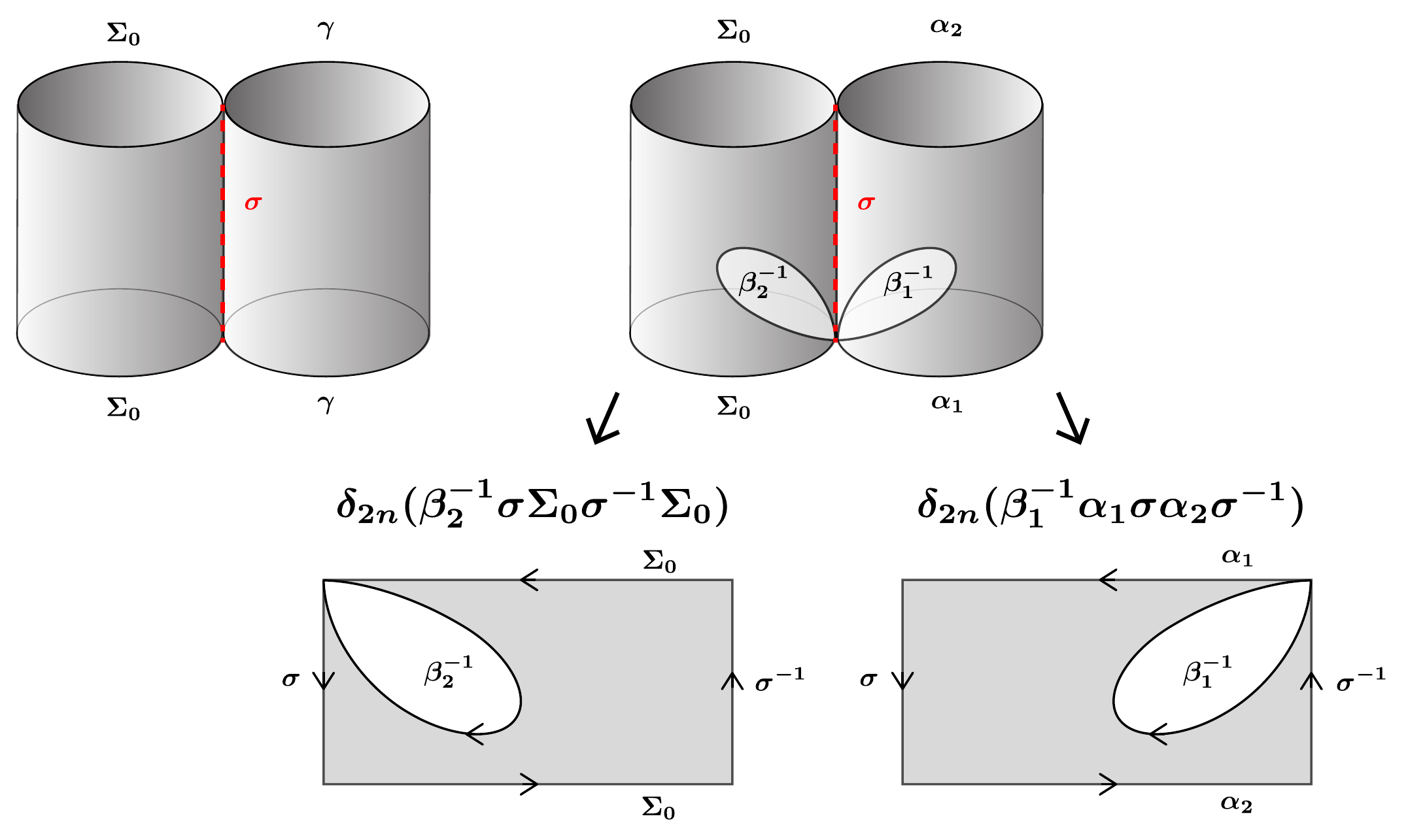}
\caption{Observables in TFT. The upper left figure is the number of states $I_{2n}$, where $\Sigma_0 , \gamma$ and the two ends of $\sigma$ are identified. The upper right figure is the two-point function $\langle \cO_{\alpha_1} \cO_{\alpha_2}\rangle$, where $\Sigma_0$ and the two ends of $\sigma$ are identified. The lower figures represent a pair of 2-cells in the two-point function.}
\label{fig:2PTFT2}
\end{center}
\end{figure}

Next we rewrite the free two-point functions \eqref{2pt formula2} to make contact with \eqref{TFT part fn1},
\begin{align}
\langle \cO_{\alpha_1} \cO_{\alpha_2}\rangle 
= \sum_{ \beta_1 \in S_{2n} } \sum_{ \beta_2 \in S_{2n} } 
 \sum_{\sigma \in S_{2n} }   N_c^{ C ( \beta_1) }N_f^{ { 1 \over 2 } {C ( \beta_2 )   } } 
\delta_{2n} ( \beta_1^{-1}  \alpha_1 \sigma \alpha_2 \sigma^{-1}) 
\delta_{2n} (\beta_2^{-1}   \sigma \Sigma_0 \sigma^{-1} \Sigma_0 ) 
\label{2pt TFT2}
\end{align} 
If we take the large $N_f, N_c $ limit, the leading terms only come from $\beta_1 = \beta_2 = 1$. 
If we set $[\alpha_1] = [\alpha_2^{-1}] = [\alpha]$ and sum over the conjugacy class $[\alpha]$, we reproduce the number of states \eqref{TFT part fn1} as
\begin{align}
\frac{1}{|S_n[S_2]|} \sum_\alpha \langle \cO_{\alpha} \cO_{\alpha}\rangle 
&= N_c^{2n} N_f^n \sum_{\alpha \in S_{2n}} \sum_{\sigma \in S_{2n} } 
\delta_{2n} (\alpha^{-1} \sigma \alpha \sigma^{-1}) 
\delta_{2n} (\sigma \Sigma_0 \sigma^{-1} \Sigma_0 ) 
\notag\\
&= N_c^{2n} N_f^n \, I_{2n} 
\end{align}
Note that the free two-point functions \eqref{2pt TFT2} become diagonal in the large $N_f, N_c $ limit owing to the symmetry \eqref{equivalence Oar}.

The formula (\ref{2pt TFT2}) can be recognised as the partition function of the $S_{2n}$ TFT on a 2-complex ${\cal M }$ with  boundaries and a defect, as we now describe.
The first delta function is associated to a 2-torus with the disc removed. The boundary of the disc has permutation  $\beta_2$, summed with  $N_f^{ C ( \beta_2 ) \over 2 }$ (this forms the $ \Omega_{N_f} $ factor).  One of the cycles of the 2-torus is constrained to be the permutation $ \Sigma_0$. This constraint can be viewed as a defect. 
The second delta function is associated to a topological quotient of a cylinder with a disc removed, related to $\beta_1$, 
which is summed with weight $N_c^{ C ( \beta_2 ) }$ to give $\Omega_{ N_c} $. The permutations $\alpha_1, \alpha_2$ correspond to the two boundary circles of the cylinder, and a point from each of the two boundary circles is identified by the quotienting. The 2-torus and the quotiented cylinder are glued along a circle, to form the $2$-complex which we call $\cM$. This is illustrated in Figure \ref{fig:2PTFT2}.

The 2-complex ${\cal M}$ cannot be the cell decomposition of a 2-manifold (which should locally be $\bb{R}^2$), because a 1-cell (the red-dashed circle denoted by $\sigma$) is incident on four 2-cells. 
Still, it should be possible to realise it as the 2-skeleton of a higher dimensional manifold. 
In that case  higher dimensional TFT would be the natural setting for the interpretation of the 2-point function. 
TFT3  has arisen in the context of refined  counting formulae for graphs in \cite{refcount}.


\section{One-loop operator mixing}\label{sec:one-loop}

In this section we will compute the mixing matrix 
under the action of the one-loop dilatation operator \cite{MZ02}
\begin{align}
H=&
-\frac{1}{2}tr[\Phi_m,\Phi_n][\check{\Phi}^m,\check{\Phi}^n]
-\frac{1}{4}tr[\Phi_m,\check{\Phi}^n][\Phi_m,\check{\Phi}^n],
\label{def:one-loop dil MZ}
\end{align} 
where 
$(\check{\Phi}^m)_{ij}(\Phi_n)_{kl}=\delta^m_n\delta_{jk}\delta_{il}$. 
On the representation basis, the mixing matrix is almost 
diagonal, where the non-zero components are explained by the repositioning of 
boxes. 

\quad 
 
On the permutation basis, the mixing matrix is given by 
\begin{align}
H \cO_{\sigma}=& \sum _{\rho}M_{\sigma,\rho}\cO_{\rho}
\\
M_{\sigma,\rho}
=&
- \sum_{\langle i,j\rangle}
\sum_{\beta \in S_{2n-1}^{\langle j\rangle}}
\delta_{2n}([\sigma,(ij)]X^{(j)}\beta^{-1})
\delta_{2n}([\rho^{-1},(ij)]\beta) \nonumber \\
&
-N_f \sum_{(i,j)}\delta_{2n}( \rho^{-1}(ij) \sigma )
-\sum_{\langle i,j\rangle}
\delta_{2n}\left(\rho^{-1} (\Sigma_0(i)j)(ij) \sigma (\Sigma_0(i)j)\right)
\nonumber \\
&
+ N_f \sum_{(i,j)}
\sum_{\beta \in S_{2n-1}^{\langle j \rangle}}
\delta_{2n}((ij)\sigma X^{(j)}\beta^{-1})
\delta_{2n}( \rho^{-1}(ij)\beta )
\nonumber \\
&
+
 \sum_{\langle i,j\rangle}
\sum_{\beta \in S_{2n-1}^{\langle j \rangle}}
\delta_{2n}((ij)\sigma X^{(j)}\beta^{-1})
\delta_{2n}(\rho^{-1}(\Sigma_0(i)j)(ij)\beta (\Sigma_0(i)j))
\label{mixing_matrix_mesonic_singlet}
\end{align} 
where $\Sigma_0=(12)(34)\cdots (2n-1,2n)$, 
and we have defined 
\begin{align}
X^{(j)}=N_c+\sum_{k(\neq j)} (kj).
\end{align} 
The sum $\sum_{(i,j)}$ is over 
$(1,2)$, $(3,4),\cdots$, while the sum $\sum_{\langle i,j\rangle}$ 
is over the other pairs. We do not distinguish $(i,j)$ and $(j,i)$.
The first line of \eqref{mixing_matrix_mesonic_singlet} comes from the first term of \eqref{def:one-loop dil MZ}, and the remaining lines come from the second term.
The derivation of this mixing matrix is presented in Appendix \ref{sec:mixing_detail}.

Given that we have expressed the mixing matrix purely in terms of permutations, 
we expect that there will be an interpretation in terms of permutation topological field theory. 
The construction of 2-complex will be analogous to our interpretation for the free field two-point function given in Figure \ref{fig:2PTFT2}, 
but will involve some new features given the additional complexity apparent here. We will return to this problem 
in the near future.

\quad

In the latter part of this section we study the mixing matrix on the 
representation basis. 
We denote the change of basis by 
\begin{align}
&P_{R,\Lambda_1,\tau}=\sum_{\alpha}c_{R,\Lambda_1,\tau}(\alpha)[\alpha], 
\nonumber \\
&[\alpha]
=\sum_{R,\Lambda_1,\tau}f_{R,\Lambda_1,\tau}(\alpha)P_{R,\Lambda_1,\tau}. 
\end{align} 
See \eqref{def:OR,Lambda1,tau} and \eqref{def:inverse transf} for the definition of $c_{R,\Lambda_1,\tau}$ and $f_{R,\Lambda_1,\tau}$.
The mixing matrix on the representation basis 
is related to the mixing matrix on the permutation basis by 
\begin{align}
M_{R,\Lambda_1,\tau}^{R^{\prime},\Lambda_1^{\prime},\tau^{\prime}}
=\sum_{\sigma,\rho\in S_{2n}}c_{R,\Lambda_1,\tau}(\sigma)M_{\sigma,\rho}
f_{R^{\prime},\Lambda_1^{\prime},\tau^{\prime}}(\rho).
\end{align} 
Let us take one term in the first line of \eqref{mixing_matrix_mesonic_singlet}, and simplify the mixing matrix on the representation basis,
\begin{gather}
\sum_{\sigma,\rho\in S_{2n}}c_{R,\Lambda_1,\tau}(\sigma)
\left(
\sum_{\beta\in S_{2n-1}}\delta_{2n}(\sigma U \beta)
\delta_{2n}(\rho V \beta^{-1})\right)
f_{R^{\prime},\Lambda_1^{\prime},\tau^{\prime}}(\rho)
\notag \\[1mm]
U = (ij) X^{(j)}, \ \ V = (ij).
\label{mixing_matrix_rep_to_native}
\end{gather}
Expand the Kronecker delta's using \eqref{delta_sigma}.
We now use the following formula to remove the sum over $\beta$ 
\begin{align}
\sum_{\beta\in S_{2n-1}}D^R_{ij}(\beta)D^{R'}_{kl}(\beta^{-1})
&=(2n-1)!\sum_{r,m,n}
\frac{1}{d_r}
B^{R,r}_{i,m }B^{R,r}_{j,n}
B^{R',r}_{k,n}B^{R',r}_{l,m}
\notag \\
&=(2n-1)!\sum_{r,m,n}
\frac{1}{d_r} \, I^{RR',r}_{il} I^{RR',r}_{jk}
\label{reduction_sum_formula_beta}
\end{align} 
Here we have introduced the branching coefficient
\begin{align}
B^{R,i}_{r,m }=\langle R,i |R\rightarrow r,m\rangle,
\end{align} 
where the $r$ is a Young diagram with $2n-1$ boxes, and 
$m$ runs over $1,\cdots ,d_r$.  
The quantity $I^{RR',r}_{il}$ is the intertwiner map 
\bea 
I^{RR',r}_{il} = \sum_m B^{R,r}_{i,m }B^{R',r}_{l,m}
\eea
of \cite{1012.3884,1108.2761}, here expressed in terms of branching coefficients. 
In the formula, 
the RHS is non-zero when both $g([1],r;R)$ and $g([1],r;R')$ are non-zero.\footnote{
$g([1],r;R)$ is the Littlewood-Richardson coefficient \eqref{def:LR coeff}.
} 
In other words, the RHS of (\ref{reduction_sum_formula_beta}) is non-zero 
only when the $R$ is obtained from the $S$ by moving a single box. 
Performing the sum over $\sigma,\rho$ in (\ref{mixing_matrix_rep_to_native}),
\begin{align}\label{rep-mix-sel-rule} 
&\sum_{\sigma,\rho}c_{R,\Lambda_1,\tau}(\sigma)
\left(
\sum_{\beta\in S_{2n-1}}
\frac{1}{((2n)!)^2}d_Ad_B
D^{A}_{ji}(\sigma U)D^A_{ij}(\beta)
D^B_{lk}(\rho V )D^B_{kl}(\beta^{-1})
\right)
f_{R^{\prime},\Lambda_1^{\prime},\tau^{\prime}}(\rho)
\nonumber \\
& = (2n)!
B^{\Lambda_1}_{k^{\prime}}
S^{\tau,} {}^{\Lambda_1}_{k'} {}^R_s {}^{R}_{j} 
D^R_{si}( U)
\left(\sum_{\beta}D^R_{ij}(\beta)D^{R^{\prime}}_{kl}(\beta^{-1})\right)
D^{R^{\prime}}_{qk}( V)
d_{R^{\prime}} 
B^{\Lambda_1^{\prime}}_{k^{\prime\prime}}
S^{\tau',} {}^{\Lambda'_1}_{k'} {}^{R'}_{l} {}^{R'}_{q} 
\end{align} 

This mixing matrix is non-zero if $R$ and $R^{\prime}$ become identical after the move of a single box. 
This kind of mixing, re-positioning of boxes, is common 
for representation bases, which has been studied concretely in \cite{0701067, 0710.5372, Brown08, BHR08,1012.3884, Kimura13}.
The diagonalisation of one-loop mixing is still not trivial, which has been achieved in the $SU(2)$ sector in special cases \cite{CdML11,1108.2761}. This will be an interesting avenue for future investigation.

\section{Conclusion and Discussion}\label{sec: conclusion}

We summarise the results of this paper. 
We studied mesonic operators, that is the $O(N_f)$-singlet scalar operators in $U(N_c)$ gauge theory, and computed free field two-point functions.
The two-point functions are expressed them in terms of permutations. 
We performed a Fourier transform from the permutation basis to the representation basis, which made the two-point functions diagonal.
Our mesonic operator provided a concrete realisation of the formula for diagonal operators \cite{BHR08}.
We counted the number of operators in both bases by applying inverse Fourier transform.
It was important to remove the redundancy of the wreath product group $S_n[S_2]$, noting that $(S_{2n}, S_n[S_2])$ is a Gelfand pair.
We computed the one-loop mixing matrix in both bases.

Our expression for the 2-point function on the permutation basis was used 
to give an interpretation in terms of TFT, based on topological lattice gauge theory of permutations equipped with defects. 
This extends the connection between TFT and quiver gauge theories \cite{quivcalc}. An interesting problem 
is to connect these results with  axiomatic TFT as discussed in \cite{MS,Turaev,SP2011,NC2016}.

An important point is that the R-symmetry group of $\cN=4$ SYM is $SO(6)$ rather than $O(6)$. 
The $SO(N_f)$-singlet operators include
\begin{equation}
\epsilon^{a_1\cdots a_{N_f}} \, tr_{N_f} ( \sigma \, \Phi_{\vec a} ), \quad
\delta^{a_1 a_2} \epsilon^{a_3\cdots a_{N_f+2}} \, tr_{N_f+2} ( \sigma \, \Phi_{\vec a} ), \quad \dots
\label{def:baryonic}
\end{equation}
which we call baryonic singlets.\footnote{The term ``baryon" refers to the flavour group, and not the colour group.} 
The mesonic operators are invariant under $O(N_f)$, while the baryonic operators are invariant under $SO(N_f)$ only.
The mesonic and baryonic operators are orthogonal at $ g_{YM}^2 = 0 $.  Since the Lagrangian of $\cN=4$ SYM does not contain the $\epsilon$ tensor of $SO(6)$, the two-point functions remain orthogonal to all orders of perturbation theory. A systematic generalisation of 
our study of correlation functions to baryonic operators is left for the future.

Our explicit construction of operators and free-field 2-point functions has been given for  $U(N_c)$ theories. 
Generalisations to other gauge groups, e.g. $SU(N_c)$ along the lines of 
\cite{CR02,dMG04,Brown07}, and $ SO(N_c)$ or $Sp( N_c)$ following \cite{CDD1301,CDD1303} will be interesting.

Our results provide a foundation for the systematic studies of planar zero modes (PZM) in the $SO(6)$ singlet sector. 
Group-theoretical counting methods for the PZM's can be developed, analogously to the construction of general mesonic singlets here. 
Often the counting formula implies the existence of a basis of PZM's labelled by permutations or representations.
An interesting application of such a formalism is to compute the non-planar anomalous dimension of the PZM's \cite{KS15}. 
The PZM's are expected to have negative anomalous dimensions by $1/N_c$ corrections, according to the results at strong coupling \cite{DMMR99,AFP00,Uruchurtu07,CEMZ14,Goncalves14} and those of conformal bootstrap \cite{DO01,ABL14,KSS15}. 
Following this paper, the general $N_f$ setup can provide a tractable approach to finite $N_f = 6$.

Another application of our results is to determine physical quantities of $\cN=4$ SYM at this level of generality, such as non-planar correlation functions \cite{DHR01,EHKS12}, partition functions \cite{SV04,GNS05,HO06,BDHO06}, and statistical properties of one-point functions via matrix product states \cite{BdLKZ15}.

The sector of  $SO(6)$ singlets offers a rich and  interesting setting to explore non-planar effects in a non-supersymmetric sector of $ \cN=4 $ SYM.

\subsubsection*{Acknowledgements}

The work of YK is supported by JSPS KAKENHI Grant Number 15K17673. 
SR is supported by STFC consolidated grant ST/L000415/1 ``String Theory, Gauge Theory
\& Duality''. RS is supported by FAPESP grants 2011/11973-4 and 2015/04030-7.
We thank Robert de Mello Koch for discussions.


\appendix 
\renewcommand{\theequation}
{\Alph{section}.\arabic{equation}}

\section{Notation and Formulae}\label{sec:formuae}
\setcounter{equation}{0}

\subsection{Notation}

We denote by $\Phi = (\Phi_a)^i_j$ a hermitian scalar field, $a=1,2, \dots , N_f$ and $i,j = 1,2, \dots N_c$. The flavour group is $SO(N_f)$ and the colour group is $U(N_c)$.
$\Phi$ belongs to the fundamental representation $V_F$ of $SO(N_f)$, and to the adjoint representation $V_C \otimes \ol{V}_C$ of $U(N_c)$. The case of $N_f=6$ describes the six scalars of $\cN=4$ SYM in four dimensions.

\bigskip
The $U(N_c)$ Wick-contraction rule is
\begin{equation}
\contraction{(}{\Phi_a}{)_i^j \, (}{\Phi_b}
(\Phi_a)_i^j \, (\Phi_b)_k^l = \delta_{ab} \, \delta^l_i \, \delta^j_k \,.
\label{UNc Wick}
\end{equation}
We introduce a gauge-covariant operator $\Phi_a$, which is related to the component fields by
\begin{alignat}{9}
\langle j_1 j_2 \dots j_{2n} \mid \Phi_{a_1} \otimes \Phi_{ a_2} \otimes \cdots \otimes  \Phi_{a_{2n}} 
\mid i_1 i_2 \dots i_{2n} \rangle =
(\Phi_{a_1})_{i_1}^{ j_{1}  } ( \Phi_{a_2} )_{ i_2}^{ j_{2} } 
\cdots ( \Phi_{ a_{2n}} )_{ i_{2n}}^{ j_{2n}}
\end{alignat}
Permutations act on the bases as 
\begin{alignat}{9}
\langle \vec j \, | \alpha &\equiv
\langle j_1 j_2 \dots j_{2n} | \alpha & &= \langle j_{\alpha^{-1}(1)} j_{\alpha^{-1}(2)} \dots j_{\alpha^{-1}(2n)} |
\\
\alpha | \vec i \,\rangle &\equiv
\alpha | i_1 i_2 \dots i_{2n} \rangle & &= | i_{\alpha(1)} i_{\alpha(2)} \dots i_{\alpha(2n)} \rangle
\end{alignat}
Thus, permutations act on SYM fields by 
\begin{equation}
\langle \vec j \mid \alpha \( \Phi_{a_1} \otimes \Phi_{ a_2} \otimes \cdots \otimes  \Phi_{a_{2n}} \) \beta 
\mid \vec i \, \rangle
= (\Phi_{a_1})_{i_{\beta(1)}}^{ j_{ \alpha^{-1}(1)}  } 
( \Phi_{a_2} )_{ i_{\beta(2)} }^{ j_{ \alpha^{-1}(2)} } 
\cdots 
( \Phi_{ a_{2n}} )_{ i_{\beta(2n)} }^{ j_{ \alpha^{-1}(2n)}} \,,
\end{equation}
and thus
\begin{equation}
\rho ( \Phi_{a_{1}}\otimes \Phi_{a_{2}} \otimes  \cdots \otimes \Phi_{a_{2n}}) \rho^{-1}
= \Phi_{a_{\rho(1)}}\otimes \Phi_{a_{\rho(2)}} \otimes  \cdots \otimes \Phi_{a_{\rho(2n)}} \,.
\label{reordering}
\end{equation}

\bigskip
The order of a finite group $G$ is denoted by $\abs{G}$. 
The words ``representations of $S_L$'', ``Young diagrams'', and ``partitions of $L$'' are used interchangeably.
A partition of $L$ is expressed in two ways. The first expression is 
\begin{equation}
Q = \rep{ q_1, q_2, \dots , q_\ell } \vdash L, \qquad 
{\tt q_1} \ge {\tt q_2} \ge \dots \ge {\tt q_\ell} \,, \qquad
\sum_{i = 1}^\ell {\tt q}_i = L .
\label{def:partition 1st}
\end{equation}
If we collect the same ${\tt q}$'s together, we obtain the second expression \eqref{def:p cycle-type}.
The symbol $c_1(Y)$ is the length of the first column of the Young diagram $Y$. 
In \eqref{def:partition 1st}, 
\begin{equation}
c_1 (Q) = \ell.
\label{def:c1Q}
\end{equation}
Clearly $c_1(Q) \ge c_2 (Q) \ge \dots \ge c_{\tt q_1} (Q)$.

\subsection{Formulae}

In this subsection the definition of group theory quantities, and 
group theory formulas are collected.

The matrix elements $(i,j)$ of the group element $\sigma \in S_L$ 
in the representation $R$ are denoted by 
$D^R_{ij}(\sigma)$. 
We assume all representations of $S_L$ are real and unitary, and thus $D^R_{ij}(\sigma) = D^R_{ji}(\sigma^{-1})$.
They satisfy the so-called grand orthogonality relation 
\begin{align}
\sum_{\sigma\in S_{L}}D^R_{ij}(\sigma)D^S_{kl}(\sigma^{-1})=
\frac{L!}{d_R}\delta_{il}\delta_{jk} \delta^{RS}, 
\label{orthogonality_representation}
\end{align}
where 
$d_R$ is the dimension of $R$ of symmetric group $S_L$, 
\begin{equation}
d_R=\frac{L!}{\prod_{i,j}h(i,j)}
\end{equation}
The product is over the boxes of the Young diagram $R$ with 
$i,j$ labelling the rows and columns. 
The quantity $h(i,j)$ is the hook-length associated with the box
at $(i,j)$, namely the number of boxes intersecting the hook which extends from $(i,j)$ toward the right and bottom. 
For example,
\begin{equation}
\Yvcentermath1
{\small \yng(4,2)} \quad\Rightarrow \quad
\small \young(5421,21) \ , \quad
h (1,1) = 5, \ h (1,2) = 4, \ h (1,3) = 2, \ h (1,4) = 1,\ \dots 
\label{def:hook-length}
\end{equation}

\quad 

The character of $\sigma$ in the representation $R$ is denoted by 
$\chi^R(\sigma)$. 
The characters satisfy the orthogonality
\begin{align}
\sum_{R\vdash L}\chi^R(\sigma)\chi^R(\rho)=
\sum_{\gamma\in S_{L}}\delta_{L}
(\gamma \sigma \gamma^{-1}\rho)
\label{formula_two_characters_combine_delta}
\end{align}
where the delta function is defined by 
\begin{eqnarray}\label{delta_sigma}
\delta_L (\sigma)=\frac{1}{L!}\sum_{R\vdash L}d_R \, \chi^R(\sigma)
=\left\{ \begin{array}{ll}
1 & (\sigma=1) \\
0 & (\sigma\neq 1) 
\end{array} \right.
\end{eqnarray}

\quad 

Consider the tensor space $V^{\otimes L}$, where $V$ is the fundamental representation of
the unitary group $U(N)$. The symmetric group acts on the tensor space 
by permuting the $L$ factors. 
From the fact that 
these two actions commute each other, 
the tensor space can be decomposed as 
\begin{align}
V^{\otimes L}=\mathop \bigoplus_R (V_R^{U(N)}\otimes V_R^{S_L})
\label{SW_dual_decomposition}
\end{align}
where the sum is over the Young diagrams with at most $N$ rows, 
which is expressed in terms of the length of the first column $c_1(R)$
by $c_1(R)\le N$.
This equation is the Schur-Weyl duality between $U(N)$ and $S_L$.

A trace over this tensor space is denoted by $tr_L$,
\begin{equation}
tr_L(\sigma) = N^{C(\sigma)} 
=N^L \delta_L(\Omega_L\sigma)
\end{equation}
where we have defined the quantity 
\begin{align}
\Omega_L=\sum_{\sigma\in S_L}\sigma N^{C(\sigma)-L}.
\end{align}
According to the Schur-Weyl duality, a trace of an 
element $\sigma \in S_L$ 
can be written as 
\begin{align}
tr_L(\sigma) 
=\sum_{R\vdash L,c_1(R)\le N} Dim(R) \, \chi^R(\sigma)
\label{Schur_Weyl_duality}
\end{align}
where $Dim(R)$ is the dimension of $R$ of Lie group $U(N)$
\begin{equation}
Dim(R) = \prod_{i,j}\frac{N-j+i}{h(i,j)}
\end{equation}
Setting $\sigma=1$ in (\ref{Schur_Weyl_duality}) gives the identity 
\begin{align}
N^L=\sum_{R\vdash L,c_1(R)\le N} Dim(R)d_R
\end{align}

\quad

Let $R_1 \,, R_2$ be the irreducible representations of $S_L$\,.
$\Clebsch{\tau}{\Lambda_1}{R_1}{R_2}{k}{i_1}{i_2}$ is the Clebsch-Gordan (CG) coefficients, defined by the irreducible decomposition of the tensor product $R_1 \otimes R_2 = \oplus S$, with the multiplicity label $\tau$, 
\begin{equation}
\ket{\Lambda,\tau,k} = \sum_{i_1, i_2 =1}^{d_R}  
\Clebsch{\tau}{\Lambda_1}{R_1}{R_2}{k}{i_1}{i_2} \, 
\ket{R_1, i_1} \otimes \ket{R_2, i_2} .
\label{def:CG coeff}
\end{equation}
The indices $(k,i_1,i_2)$ specifies the elements of $(R,R_1,R_2)$. 
The multiplicity label $\tau$ runs over $1,\cdots ,C(R,R,\Lambda)$, where $C(R,S,T)$ is called the CG number (also known as CG multiplicity or Kronecker coefficient)
\begin{equation}
C(R,S,T) = \frac{1}{|S_{L}|} 
\sum_{\sigma \in S_{L}} \, \chi^R (\sigma) \chi^{S} (\sigma) \chi^T (\sigma).
\label{def:CG number}
\end{equation}
The CG coefficients satisfy the following properties
\begin{align}
\sum_{\sigma\in S_L}
D^{\Lambda_1}_{ba}(\sigma)
D^R_{jk} (\sigma)
D^R_{li}( \sigma^{-1})
=\frac{L!}{d_{\Lambda_1}}\sum_{\tau}
\Clebsch{\tau}{\Lambda_1}{R}{R}{b}{i}{j} \, 
\Clebsch{\tau}{\Lambda_1}{R}{R}{a}{l}{k}
\label{BHR_163}
\end{align}
and 
\begin{align}
\sum_{ij}
\Clebsch{\tau}{\Lambda_1}{R}{R'}{a}{i}{j} \,
\Clebsch{\tau'}{\Lambda'_1}{R}{R'}{b}{i}{j}
= \delta^{\Lambda_1 \Lambda_1^{\prime}} \, \delta^{\tau \tau^{\prime}} \delta_{ab} \,, \qquad
\sum_{\tau, \Lambda_1, a} 
\Clebsch{\tau}{\Lambda_1}{R}{R'}{a}{i}{j} \,
\Clebsch{\tau}{\Lambda_1}{R}{R'}{a}{k}{l}= \delta_{ik} \, \delta_{jl}
\label{BHR_160}
\end{align}

\quad 

Let $R_1 , R_2 , R$ be the irreducible representations of $S_m , S_n , S_{m+n}$, respectively.
The Littlewood-Richardson coefficient $g(R_1,R_2;R)$ counts the number of $R_1\otimes R_2$ appearing in the decomposition of $R$ under $S_m\times S_n$, 
\begin{equation}
g(R_1,R_2;R)=
\frac{1}{m! \, n!}
\sum_{\sigma_1\in S_m,\sigma_2 \in S_n}
\chi_{R_1}(\sigma_1^{-1})\chi_{R_2}(\sigma_2^{-1})\chi_{R}(\sigma_1 \circ \sigma_2).
\label{def:LR coeff}
\end{equation}

\subsection{Gelfand pair and coset type}\label{app:Gelfand}

A pair of finite groups $(G,H)$ with $G \supset H$ is called Gelfand pair if they satisfy either of the following conditions:
\begin{itemize}
\item[(i)]
Any irreducible representation of $G$ contains at most one singlet representation of $H$. 
\item[(ii)]
Consider a set of functions on the double coset $H \backslash G/H$,
\begin{equation}
\cC (G,H) := \Bigl\{ w : G \to G \,\Big|\, w (g) = w ( \gamma_1 g \gamma_2), \ \ \ (g \in G, \ \gamma_1, \gamma_2 \in H) \Bigr\}
\label{def:CGH}
\end{equation}
where the multiplication is defined by convolution. 
The algebra $\cC(G,H)$ is commutative.
\end{itemize}
The two conditions are equivalent \cite{MacdonaldBook}. The Gelfand pair for Lie groups is defined similarly.

An important example of the Gelfand pair is $(S_{2n}, S_n[S_2])$. 
From (i), an irreducible representation $R$ of $S_{2n}$ satisfies
\begin{align}
M^{R}_{1_{S_{n}[S_2]}}
:=
\frac{1}{|S_n[S_2]|}
\sum_{\gamma\in {S_{n}[S_2]}
}\chi^{R}(\gamma)
=\left\{ \begin{array}{ll}
1 & \ \ (\mbox{$R$ is even}) \\
0 & \ \ (\mbox{$R$ is odd}) 
\end{array} \right. .
\label{mult-free SnS2}
\end{align}
An even representation $R$ corresponds to even Young diagram
\begin{equation}
R = \rep{2r_1,2r_2,\cdots} \vdash 2n, \qquad
{\tt r_i} \in \bb{Z}_{>0}
\label{def:even rep}
\end{equation}
in \eqref{def:partition 1st}.

The double coset function in the condition (ii) corresponds to $W(\sigma)$ defined in \eqref{def:Wsigma}, which satisfies the equivalence relation,
\begin{equation}
W(\sigma)=W(\gamma_1 \sigma \gamma_2) , \quad 
(\sigma\in S_{2n}, \ \ \gamma_1,\gamma_2\in S_n[S_2]).
\end{equation}
The function $W(\sigma)$ defines the unique {\it coset type} for 
each permutation $\sigma$, 
\begin{equation}
\tilde p = [2^{\tilde p_1}, 4^{\tilde p_2}, \dots, (2n)^{\tilde p_{n}}], \qquad \sum_{i=1}^{n} i \, \tilde p_i = n \,.
\label{def:tp coset-type}
\end{equation}
where 
$\tilde{p}_l$ in \eqref{def:tp coset-type} is the number of length-$2l$ loops in $W(\sigma)$ shown in Figure \ref{fig:W(sigma)2}.\footnote{A loop $W(1)$ at $n=2$ is defined to have length 2.}
The number of all loops is equal to the power of $N_f$ in 
(\ref{W_sigma_brauer}), 
\begin{equation}
z(\sigma)=\sum_{l} \tilde p_l.
\end{equation}
Two permutations $\sigma_1$, $\sigma_2 \in S_{2n}$ have the same coset type
if and only if they are related by  
$\sigma_1= \gamma_1 \sigma_2\gamma_2$ for $\gamma_1,\gamma_2\in S_n[S_2]$.


\section{Powers of $N_f$}\label{sec:powers Nf}
\setcounter{equation}{0}

We denote the number of loops in the graph of $W(\sigma)$ by $z(\sigma)$.
We will show the identity
\begin{equation}
z(\sigma) =\frac{1}{2}C(\Sigma_0\sigma^{-1} \Sigma_0\sigma) ,
\label{Wsig identity}
\end{equation}
where $\Sigma_ 0 = (12)(34) \dots (2n-1,2n )$ and $C(\sigma)$ counts the number of cycles in $\sigma \in S_{2n}$\,.

\begin{figure}
\begin{center}
\includegraphics[scale=0.7]{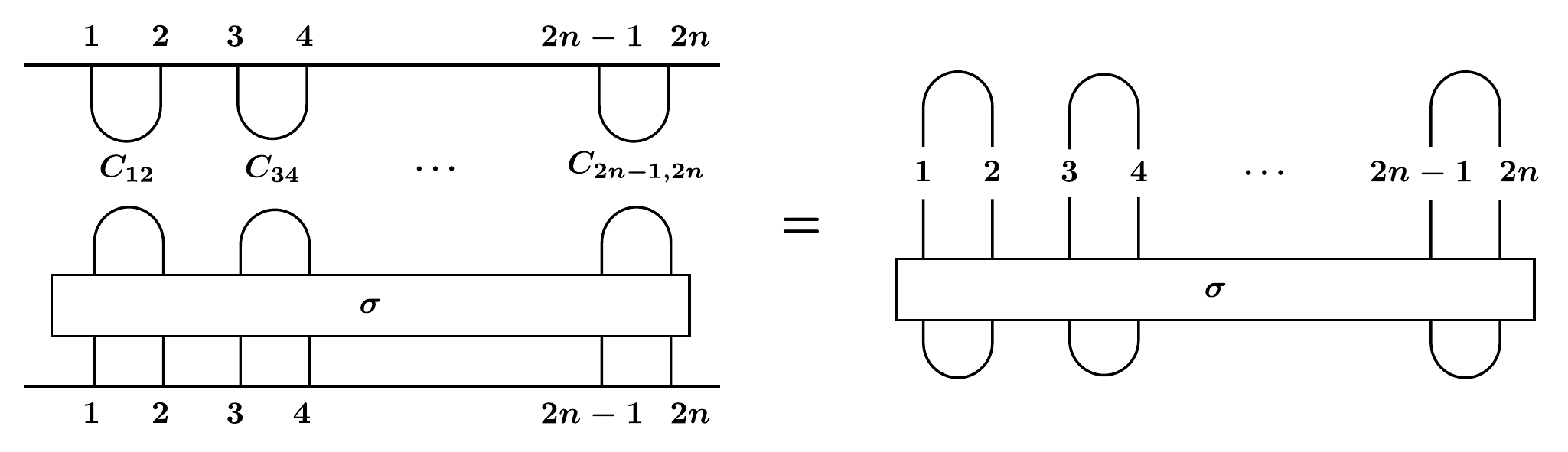}
\caption{This is identical to Figure \ref{fig:W(sigma)}. The left figure shows $C_{12}C_{34} \cdots C_{2n-1,2n} \sigma$ acting on $V_F^{ \otimes 2n }$. The upper and lower horizontal lines are identified when taking the trace, as in the right figure.}
\label{fig:W(sigma)2}
\end{center}
\end{figure}

The structure of $W(\sigma)$ is depicted in Figure \ref{fig:W(sigma)2}.
Consider the labelled points $\{ 1 , 2, \cdots , 2n \}$ in that diagram. If these points move along the lines toward the upper arcs, they will return to the labelled points $\{ 1 , 2, \cdots , 2n \}$ after a permutation 
\bea 
&& \Sigma_0 = (12) (34) \cdots (2n-1~  2n ) 
\eea
If the points went down through $ \sigma $, the lower arcs, and back up $ \sigma$ again, they will undergo the permutation 
\bea 
\Sigma_1 \equiv \sigma^{-1} \Sigma_0 \sigma 
= (\sigma(1) \sigma(2)) (\sigma(3) \sigma(4)) \cdots (\sigma(2n-1) \sigma(2n)).
\eea

The graph of $W(\sigma)$ contains flavour loops. A flavour loop is a sequence of transitions of the form
\begin{equation}
\Sigma_0 \Sigma_1 \,, \quad
\Sigma_0 \Sigma_1 \Sigma_0 \Sigma_1 \,, \quad
\Sigma_0 \Sigma_1 \Sigma_0 \Sigma_1 \Sigma_0 \Sigma_1 \,, \quad
\cdots 
\label{sequences S0S1}
\end{equation}
or their inverses. Note that $\Sigma_0^2 = \Sigma_1^2 =1$. We cannot return to the original point by after an odd number of $\Sigma$'s, because both $\Sigma_0$ and $\Sigma_1$ have cycle type $[2^n]$. In other words, $\Sigma_0$ and $\Sigma_1$ acting on $i$ behaves as a permutation of odd signature for any $i$.
The number of flavour loops, i.e. the power of $N_f$ denoted by $z ( \sigma )$, is same as the number of orbits in the subgroup of $S_{2n}$ generated by $\langle \Sigma_0 , \Sigma_1 \rangle$. 
This is also the number of connected components in the ribbon graph determined by the permutations  $ \Sigma_0 , \Sigma_1$.

\bigskip
Let us introduce the notation
\begin{equation}
2i-1 = i^-, \qquad 2i = i^+, \qquad (i=1,2, \dots, n ) .
\end{equation}
Recall that $S_{n }[S_2]$ is the stabiliser of $\Sigma_0$\,,
\begin{equation}
\xi^{-1} \Sigma_0 \, \xi = \Sigma_0 \,, \qquad \forall \xi \in S_{n}[S_2].
\end{equation}
By using this ``gauge degree of freedom'' of $S_n[S_2]$, we may transform $\Sigma_1$ to $\Sigma'_1 = \xi^{-1} \Sigma_1 \, \xi$ without changing $C(\Sigma_0 \Sigma_1)$. There exists a useful gauge:

\begin{lem}\label{lem:can S1}
By a gauge transformation in $S_n[S_2]$, we can transform $\Sigma_1$ to the form
\begin{equation}
\Sigma'_1 = (1^- \, \tau(1)^+) \dots ( n^- \, \tau(n )^+), \qquad
\tau \in S_n \,.
\label{Sigma1 canonical}
\end{equation}
\end{lem}

\begin{proof}[Proof of \ref{lem:can S1}]
We draw another graph of $W(\sigma)$ emphasising the structure of loops, with $1^+, 2^+, \dots$ along the upper line and $1^-, 2^-, \dots$ along the lower line. We connect the points $i^\pm$ and $j^\pm$ when $(i^\pm, j^\pm)$ belong to $\Sigma_0$ or $\Sigma_1$ as shown in Figure \ref{fig:NH}.
The horizontal edges of $\Sigma_1$\,, namely those connecting $(+,+)$ or $(-,-)$, will be called H-edges. The other edges of $\Sigma_1$\,, $(+,-)$ or $(-,+)$, will be called N-edges.

\begin{figure}[t]
\begin{center}
\includegraphics[scale=0.7]{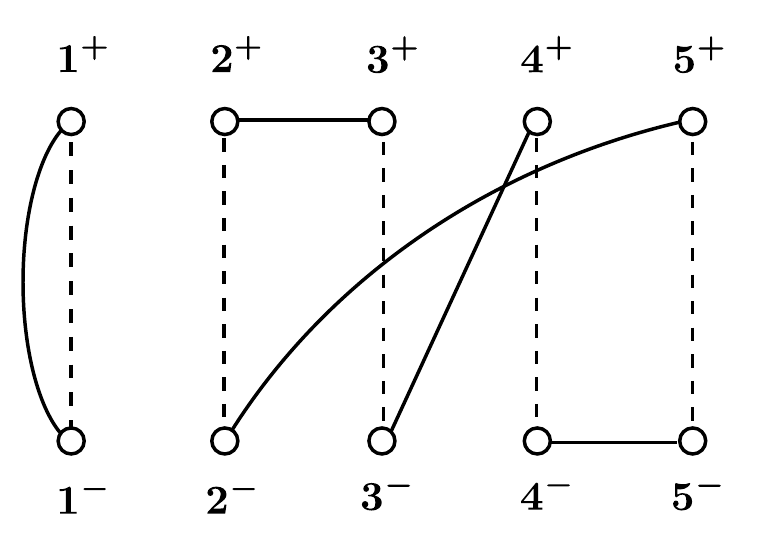}
\caption{Another graph of $W(\sigma)$. The dashed edges represent the elements of $\Sigma_0$\,, and the sold ones those of $\Sigma_1$\,. This graph contains two loops, $(N \Sigma_0)$ and $(H \Sigma_0 N \Sigma_0 H \Sigma_0 N \Sigma_0)$.}
\label{fig:NH}
\end{center}
\end{figure}

Let us prove that every loop has to have an even number of H-edges.
As discussed in \eqref{sequences S0S1}, every loop consists of even number of edges. Within the loop, only $\Sigma_0$ and N-edges change the parity $\pm$. If we circle around the loop, then there should be no change in parity. Thus, every loop $( \Sigma_1 \Sigma_0 \dots \Sigma_1 \Sigma_0 )$ satisfies
\begin{equation}
{\rm Parity} \, ( \Sigma_1 \Sigma_0 \dots \Sigma_1 \Sigma_0 ) = 
(-1)^{\# (N)} (-1)^{\# (\Sigma_0)} = +1.
\label{parity rule}
\end{equation}
Since $\# (H) + \# (N) + \# (\Sigma_0)$ is even for each loop, the number of H-edges is also even.

The statement \eqref{Sigma1 canonical} is equivalent to saying that we can remove all H-edges by gauge transformations $S_n [S_2]$. Consider how the flip $(i^- i^+) \in S_n [S_2]$ acts on the $\Sigma_1$ edges connected to the points $i^\pm$. Inspecting Figure \ref{fig:NH}, we find
\begin{equation}
(i^- i^+) \, : \, 
\begin{cases}
H (\Sigma_0)_{i^\pm} H \ \to \ N (\Sigma_0)_{i^\pm} N
\\
H (\Sigma_0)_{i^\pm} N \ \to \ N (\Sigma_0)_{i^\pm} H
\\
N (\Sigma_0)_{i^\pm} H \ \to \ H (\Sigma_0)_{i^\pm} N
\end{cases} .
\end{equation}
The flip cannot change $H \Sigma_0 N$ into $N \Sigma_0 N$ or $H \Sigma_0 H$, because it violates the parity rule \eqref{parity rule}. The same is true for $N \Sigma_0 H$.
Now, if we have a flavour loop $(H \Sigma_0 N \cdots \Sigma_0 H \dots )$, flipping all the edges in $\Sigma_0$ between two H-edges will convert the loop to $(N \Sigma_0 N \cdots \Sigma_0 N \dots )$. Since the number of H-edges is even, applying this process repeatedly will remove all H-edges. This means that there is a gauge transformation which converts $\Sigma_1$ to $ \Sigma_1'$ of the form (\ref{Sigma1 canonical}). 
\end{proof}

The permutation  $\tau $ is itself defined up to conjugation in $S_n$. This is in fact a 
way to understand the correspondence between partitions of $n$ and the double coset 
space $  S_n[ S_2] \setminus  S_{2n} / S_n [ S_2 ]  $, as we will explain subsequently.

Let us denote by $(\ell_1, \ell_2, \dots)$ the number of edges in the loops of the graph $W(\sigma)$. These $\{ \ell_i \}$ are all even, and satisfy $\sum_i \ell_i = 2n$. Thus, $\lambda_i \equiv \ell_i/2$ defines a partition of $n$. This partition $\lambda$ is same as the cycle decomposition of $\tau$.
We are going to relate the number of loops with the number of cycles in $\Sigma_{\rm tot} \equiv \Sigma_0 \Sigma'_1$.

\begin{cor}\label{cor: Stot}
$\Sigma_{\rm tot}$ maps minus variables to minus variables, plus to plus.
\end{cor}

\begin{proof}[Proof of \ref{cor: Stot}]
From \eqref{Sigma1 canonical} we find $\Sigma_{\rm tot} (i^-) = \Sigma_0 (\tau(i)^+) = \tau (i)^-$ and $\Sigma_{\rm tot} (i^+) = \Sigma_0 (\tau^{-1} (i)^-) = \tau^{-1} (i)^+$.
\end{proof}

Thus, $\Sigma_{\rm tot}$ splits into two disjoint actions $\Sigma_{\rm tot}^- \times \Sigma_{\rm tot}^+$\,, where $\Sigma_{\rm tot}^\pm$ acts on the set $V_n^\pm = (1^\pm, \dots, n^\pm)$. 
Then, the number of cycles is equal to
\begin{equation}
C(\Sigma_0 \Sigma_1) = C(\Sigma_{\rm tot}) = C(\tau) + C(\tau^{-1}) = 2 C(\tau).
\label{CSigTot}
\end{equation}

\begin{lem}\label{lem:conn tau}
The number of loops in $\Sigma_{\rm tot}^-$ acting on $V_n^-$ is equal to $C(\tau)$, and similarly for $\Sigma_{\rm tot}^+$ acting on $V_n^+$\,.
\end{lem}

\begin{proof}[Proof of \ref{lem:conn tau}]
Let us define
\begin{equation}
\Sigma_{\rm tot} = \Sigma_{\rm tot}^- \, \Sigma_{\rm tot}^+, \qquad
\Sigma_{\rm tot}^- \equiv \prod_{i=1}^n (i^- \tau(i)^-), \qquad
\Sigma_{\rm tot}^+ \equiv \prod_{i=1}^n (i^+ \tau^{-1} (i)^+).
\end{equation}
We can express the number of loops in $\Sigma_{\rm tot}^\pm$ as
\begin{equation}
N_f^{\# ({\rm loops})} = \prod_{h=1}^n \delta^{c_{h^-}}_{c_{\tau (h)^-}} 
= \prod_{h=1}^n \delta^{c_{h^+}}_{c_{\tau (h)^+}} = N_f^{C(\tau)},
\label{NumLoops}
\end{equation}
showing that $\# ({\rm loops}) = C(\tau)$. It can also be derived graphically as in Figure \ref{fig:SigPm}.
\end{proof}

\begin{figure}[t]
\begin{center}
\includegraphics[scale=0.7]{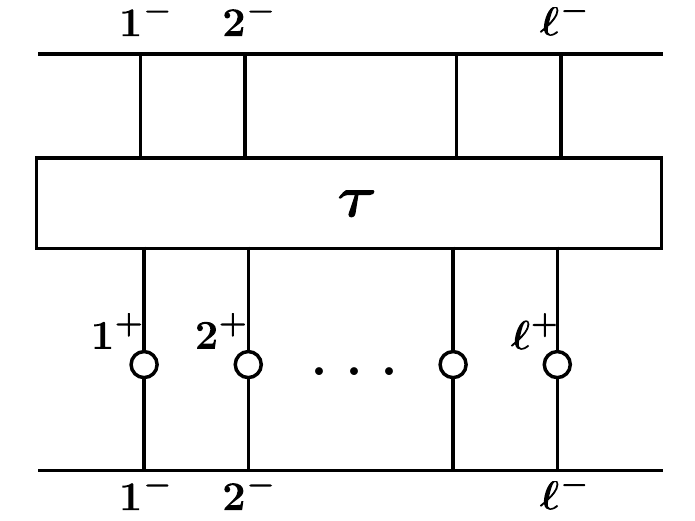}
\caption{$\Sigma_{\rm tot}^-$ acting on $(1^- \dots n^-)$. $\Sigma_0$ interchanges $(i^- i^+)$ and $\Sigma'_1$ permutes by $\tau$. We can draw a similar graph for $(1^+ \dots n^+)$.}
\label{fig:SigPm}
\end{center}
\end{figure}

The identity \eqref{Wsig identity} follows from \eqref{NumLoops} and \eqref{CSigTot}, 
\begin{equation}
W ( \sigma ) = N_f^{z(\sigma)} = N_f^{\# ({\rm loops})} = N_f^{C(\tau)} , \qquad
z(\sigma) =\frac{1}{2}C(\Sigma_0\sigma^{-1} \Sigma_0\sigma) .
\end{equation}

\bigskip
The above discussion gives a concrete insight into the coset type \eqref{def:tp coset-type};
A coset type is a partition of $n$ which parametrises the elements of double coset $S_n [S_2] \backslash  S_{2n} / S_n[S_2 ] $.
Lemma \ref{lem:can S1} says that $ \Sigma_1 = \sigma^{-1} \Sigma_0 \sigma$ can be gauge transformed by $ \xi \in S_n [ S_2 ] $ to the form \eqref{Sigma1 canonical}. 
We take $\tilde \tau$ to be a permutation in $S_{2n}$ which leaves $i^-$ fixed and acts nontrivially on $i^+\  (i=1,2, \dots n)$.
The gauge transformation of $\xi$ can be written as
\begin{align}
\xi^{-1} \sigma^{-1} \Sigma_0 \sigma \xi  
&= \tilde \tau^{-1} \Sigma_0 \tilde \tau
\label{gauge transform tau} \\
&= (\tilde \tau (1^-) \tilde \tau (1^+)) (\tilde \tau (2^-) \tilde \tau (2^+)) \cdots (\tilde \tau (n^-) \tilde \tau (n^+))
\notag \\
&\equiv ( 1^-  \tau( 1 )^+ ) ( 2^- \tau(2)^+  ) \cdots  (n^- \tau(n)^+ ). 
\notag 
\end{align}
So $\tilde \tau \xi^{-1} \sigma^{-1} $ is in the stabiliser of $ \Sigma_0$, and $\tilde \tau \xi^{-1} \sigma^{-1} = \eta \in S_n [ S_ 2 ] $ for some $ \eta$. 
Hence any $ \sigma $ can be written as 
\bea 
\tilde \tau = \eta \sigma \xi, \qquad
\eta, \xi \in S_n[S_2]
\eea
Therefore, the elements of the double coset $ S_n [S_2] \backslash S_{2n} / S_n[S_2 ] $ correspond to the permutations $\tilde \tau \in S_{2n}$, or the permutations $ \tau \in S_n$.

The condition \eqref{gauge transform tau} does not completely fix the gauge. The residual gauge freedom is conjugation of $\tau$ by $\xi \in S_n \subset S_n[S_2]$, which should not change the double coset element. 
As a result, the double coset elements are in 1-1 correspondence with the conjugacy classes in $S_n$, i.e. partitions of $n$ called coset types.


\section{Diagonal two-point functions}\label{diagonal_twopoint_functions}
\setcounter{equation}{0}

In this Appendix, we will derive (\ref{two-pt_rep_basis}) and compute the normalisation factor.

\subsection{Proof of diagonality}\label{diagonal_twopoint_functions1}

Let us first rewrite the two-point functions of the permutation basis
\begin{align}
\langle \cO_{\alpha_1}\cO_{\alpha_2}\rangle 
&=\sum_{\sigma\in S_{2n}} 
W(\sigma) \, tr_{2n}(\alpha_1 \sigma \alpha_2 \sigma^{-1}).
\label{2pt_O_alpha_to_compute}
\end{align}
The colour factor can be expanded using (\ref{Schur_Weyl_duality}) as 
\begin{align}
tr_{2n}(\alpha_1 \sigma\alpha_2 \sigma^{-1})
=&\sum_R Dim(R) \, \chi^R(\alpha_1 \sigma \alpha_2 \sigma^{-1})
\nonumber \\
=&
\sum_R Dim(R) D^R_{ij} (\alpha_1)D^R_{jk}(\sigma) D^R_{kl}(\alpha_2)
D^R_{li}( \sigma^{-1}),
\label{SWD}
\end{align}
where $Dim(R)$ is the dimension of $R$ associated with $U(N_c)$.
The flavour factor can be written  
from (\ref{W_sigma_Omega_delta}) as 
\begin{align}
W(\sigma)
=\frac{N_f^{n}}{(2n)!}\sum_{\Lambda_1}d_{\Lambda_1}
D^{\Lambda_1}_{ab}(\Omega_{2n}^{(f)})D^{\Lambda_1}_{ba}(\sigma).
\end{align}
The sum over $\sigma$ in  
(\ref{2pt_O_alpha_to_compute})
can be removed 
using  (\ref{BHR_163})
to obtain 
\begin{align}
\langle \cO_{\alpha_1}\cO_{\alpha_2}\rangle
=&
\sum_{R,\Lambda_1}\sum_{i,j,k,l,a,b} Dim(R) D^R_{ij}(\alpha_1)  D^R_{kl}(\alpha_2)
N_f^{n}D^{\Lambda_1}_{ab}(\Omega_{2n}^{(f)})
\sum_{\tau}
\Clebsch{\tau}{\Lambda_1}{R}{R}{b}{i}{j} \, 
\Clebsch{\tau}{\Lambda_1}{R}{R}{a}{l}{k}
\label{naive_basis_two-pt}
\end{align}

We now compute the two-point functions of the representation basis
\begin{align}
\langle \cO^{R,\Lambda_1,\tau}
O^{S,\Lambda_1^{\prime},\tau^{\prime}}
\rangle
=
B^{\Lambda_1}_k B^{\Lambda_1^{\prime}}_{k^{\prime}}
\Clebsch{\tau}{\Lambda_1}{R}{R}{k}{i}{j} \, 
\Clebsch{\tau'}{\Lambda'_1}{R'}{R'}{k'}{i'}{j'}
\sum_{\alpha_1,\alpha_2\in S_{2n}}
D^R_{ij}( \alpha_1^{-1})D^{R^{\prime}}_{i^{\prime}j^{\prime}}( \alpha_2^{-1})
\langle \cO_{\alpha_1}
\cO_{\alpha_2}
\rangle.
\end{align}
Substituting (\ref{naive_basis_two-pt}) into this and 
using the relation (\ref{BHR_160}),  
\begin{align}
\langle \cO^{R,\Lambda_1,\tau}
O^{S,\Lambda_1^{\prime},\tau^{\prime}}
\rangle
=
\delta_{RS}\delta_{\tau\tau^{\prime}}
\delta_{\Lambda_1\Lambda_1^{\prime}}
\left(\frac{(2n)!}{d_R}\right)^2 Dim(R) 
N_f^{n}
\langle \Lambda_1 \rightarrow 1_{S_{n}[S_2]} \mid
\Omega_{2n}^{(f)}
\mid 
\Lambda_1 \rightarrow 1_{S_{n}[S_2]} \rangle.
\label{2pt branching}
\end{align}
The last factor has several expressions 
\begin{align}
\langle \Lambda_1 \rightarrow 1_{S_{n}[S_2]} \mid
\Omega_{2n}^{(f)}
\mid 
\Lambda_1 \rightarrow 1_{S_{n}[S_2]} \rangle
&= B^{\Lambda_1}_{k_1} B^{\Lambda_1}_{k_2} D^{\Lambda_1}_{k_1 k_2} ( \Omega_{2n}^{(f)} ) 
\nonumber \\
&=
\chi_{\Lambda_1}( \Omega_{2n}^{(f)} \, p_{1_{S_{n}[S_2]}} )
\nonumber \\
&=
\omega_{\Lambda_1/2}(\Omega_{2n}^{(f)}). 
\label{def:flavour norm}
\end{align}
The second equality comes from $\gamma_1 \Omega_{2n}^{(f)} \gamma_2 =\Omega_{2n}^{(f)} $ 
for $\gamma_1,\gamma_2 \in S_{n}[S_2]$.
It is non-zero only for the case
$\Lambda_1,\Lambda_1^{\prime}$ are even Young diagrams.

\subsection{Twisting Wick-contraction rules}\label{app:zonal}

We derive the formula \eqref{formula_character_omega_flavour} by developing  the connection with \cite{MacdonaldBook}. 
The function $\omega_{\Lambda_1/2}(\Omega_{2n}^{(f)})$ which appears in the normalisation of the diagonal two-point functions \eqref{def:flavour norm} is equal to the zonal spherical function of the Gelfand pair $(GL(N_f), O(N_f))$, introduced in \cite{MacdonaldBook}.\footnote{The colour factor $Dim(R)$ in \eqref{Schur_Weyl_duality} is replaced by the Schur polynomial of the eigenvalues of $Y^2$ \cite{MacdonaldBook}.}
In developing this connection, it  is instructive to introduce a  twist  of the  two-point functions parametrised by matrices 
$T_{ab}\in GL(N_f) $ and $Y^i_j \in GL(N_c)$.

Since $P^{R,\Lambda_1,\tau}$ does not depend explicitly on $N_c$ or $N_f$, our construction of diagonal operators can be readily generalised to the case where the Wick-contraction rules are twisted,
\begin{equation}
\contraction{(}{\Phi_a}{)_i^j \, (}{\Phi_b}
(\Phi_a)_i^j \, (\Phi_b)_k^l = T_{ab} \, Y^l_i \, Y^j_k
\label{UNc Wick twisted}
\end{equation}
The two-point functions \eqref{mtr Wick formula} become
\begin{align}
\langle \cO_{\alpha_1} \cO_{\alpha_2} \rangle 
&= \delta_{\vec a} \, \delta^{\vec b} \sum_{\sigma \in \cS_{2n} } \prod_{k=1}^{2n} T^{a_k}_{b_{\sigma (k)}} 
Y^{i_{\sigma(k)}}_{j_{\alpha_2 (k)}} Y_{i_{\alpha_1 \sigma (k)}}^{j_k}
= \sum_{\sigma \in \cS_{2n} }
\( \delta_{\vec a} \, \delta^{\vec b} \, \prod_{k=1}^{2n} T^{a_k}_{b_{\sigma(k)}} \) 
\( \prod_{k=1}^{2n} (Y^2)^{i_k}_{i_{\alpha_1 \sigma \alpha_2 \sigma^{-1} (k)}} \)
\notag \\
&\equiv \sum_{\sigma \in \cS_{2n} } W_T (\sigma) P_Y (\alpha_1 \sigma \alpha_2 \sigma^{-1}),
\label{mtr Wick formula twists}
\end{align}
where $W_T$ collects the flavour factors and $P_Y$ the colour factor. They can be written as the power sum,
\begin{align}
P_Y (\rho) = \prod_{\ell=1}^{2n} tr (Y^{2\ell})^{p_\ell}, \qquad
W_T (\sigma) = \prod_{k=1}^n \pi_{k}(T^T T)^{\tilde p_k} , \qquad
\pi_k (X) \equiv tr (X^k) 
\label{W_sigma_T brauer}
\end{align}
where $p \vdash 2n$ is the cycle type of $\rho = \alpha_1 \sigma \alpha_2 \sigma^{-1}$, and $\tilde p \vdash n$ is the coset type of $\sigma$. When we untwist as $T=Y=1$, both factors reduce to
\begin{align}
P_Y (\alpha_1 \sigma \alpha_2 \sigma^{-1}) \big|_{Y=1} &= \prod_{i=1}^{2n} N_c^{p_i} = N_c^{C(\alpha_1 \sigma \alpha_2 \sigma^{-1})}
\\
W_T (\sigma) \big|_{T=1} &= \prod_{k=1}^n N_f^{\tilde p_k} = N_f^{z(\sigma)} \,.
\end{align}
A graphical representation of $W_T(\sigma)$ is shown in Figure \ref{fig:W(sigma T)}.

\begin{figure}
\begin{center}
\includegraphics[scale=0.7]{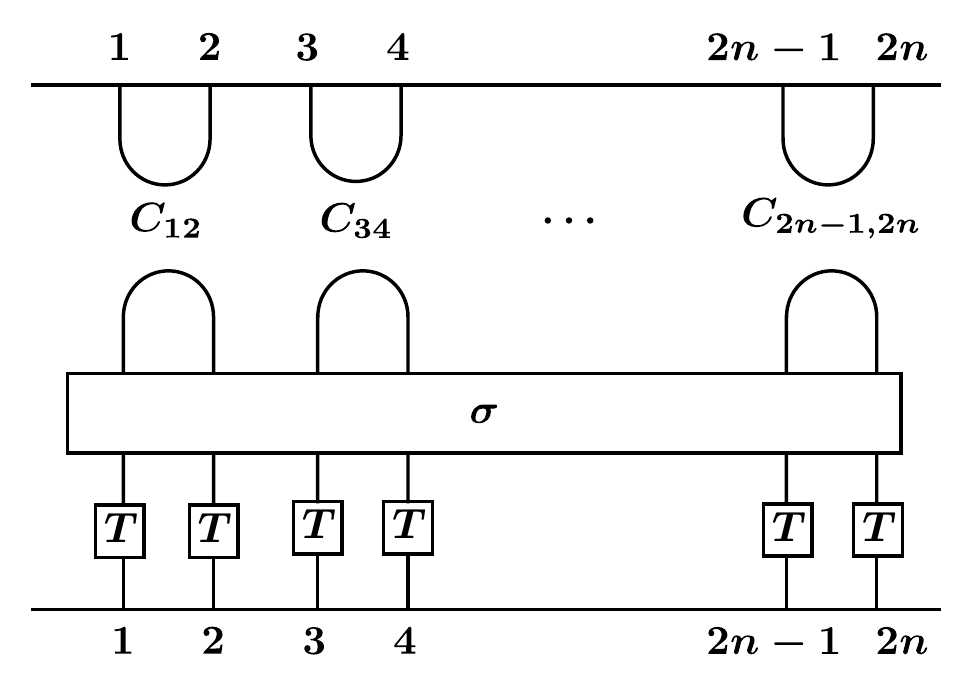}
\caption{Generalisation of $W(\sigma)$ in Figure \ref{fig:W(sigma)} including the twist $T$.}
\label{fig:W(sigma T)}
\end{center}
\end{figure}

Owing to \eqref{W_sigma_T brauer}, the operator $\Omega_{2 n}^{(f)}$ in \eqref{def:Omega factors} is twisted as
\begin{equation}
\Omega_{2 n}^{(f_T)}=\frac{1}{N_f^{n}}\sum_{\sigma \in S_{2n}} \( \prod_{k=1}^n \pi_{k} (T^T T)^{\tilde p_k} \) \sigma^{-1} .
\end{equation}
The flavour factor in the diagonal two-point functions in \eqref{two-pt_rep_basis} becomes
\begin{align}
\omega_{\Lambda_1/2} (\Omega_{2n}^{(f_T)}) 
&= \chi_{\Lambda_1} \(\Omega_{2n}^{(f_T)} p_{1_{S_n[S_2]}} \)
= \frac{1}{N_f^{n}}\sum_{\sigma \in S_{2n}} \( \prod_{k=1}^n \pi_{k} (T^T T)^{\tilde p_k} \) 
\chi_{\Lambda_1} \Bigl( \sigma p_{1_{S_n[S_2]}} \Bigr),
\notag \\
&= \abs{S_n[S_2]} \, Z_{\Lambda} (T^T T)
\label{twisted omega factor}
\end{align}
where we used $\chi_{\Lambda_1} \( \sigma^{-1} \) = \chi_{\Lambda_1} \( \sigma \)$.
It turns out that the function $\omega_{\Lambda_1/2} (\Omega_{2n}^{(f_T)}) $ is identical, up to normalisation, to the zonal spherical function $Z_{\Lambda} (T^T T)$ of the Gelfand pair $(GL(N_f), O(N_f))$ introduced in \cite{MacdonaldBook}. In \cite{MacdonaldBook} it is shown that
\begin{equation}
Z_{\Lambda} (1) = \prod_{(i,j) \in \Lambda} (N_f + 2j-i-1),
\label{zonal sph formula}
\end{equation}
where $(i,j)$ specifies the position of the Young tableau $\Lambda$. From \eqref{zonal sph formula}, we reproduce the formula \eqref{formula_character_omega_flavour}.

The overall normalisation of $\omega_{\Lambda_1/2} (\Omega_{2n}^{(f_T)})$ is determined as follows.
If we take $T=1$ and $N_f \gg 1$ in \eqref{twisted omega factor}, the leading terms come from $\sigma \in S_n[S_2]$, whose coset type is $[2^n]$. Since $(S_{2n}, S_n[S_2])$ is a Gelfand pair, the restriction to $S_n[S_2]$-invariant subspace is multiplicity-free; $\chi_{\Lambda_1} ( p_{1_{S_n[S_2]}} ) = 1$.


\section{Examples of diagonal operators}\label{app:example}

We will explain how to construct the diagonal operators \eqref{def:OR,Lambda1,tau} in $U(N_c)$ theories
\begin{align}
\cO^{R,\Lambda_1,\tau} = \sum_{k=1}^{d_{\Lambda_1}} B^{\Lambda_1}_k \, \sum_{\alpha \in S_{2n}} \sum_{i,j=1}^{d_R} 
\Clebsch{\tau}{\Lambda_1}{R}{R}{k}{i}{j}
D^R_{ij} (\alpha^{-1}) \, \cO_{\alpha} \,, 
\label{def:OR,Lambda1,tau2}
\end{align}
and give explicit examples at $2n=2,4$.

\subsection{Generality}\label{app:generality}

Let us explain our strategy.
First, we classify all irreducible representations $R$ such that $R \otimes R$ contains an even partition $\Lambda_1$\,. 
The irreducible decomposition of $R \otimes R$ can be computed from the character table by using \eqref{def:CG number}.

Second, we specify an orthonormal basis of the irreducible representations of $S_{2n}$ explicitly. We use the Young-Yamanouchi orthonormal form for this purpose \cite{AudenaertNote, HamermeshBook}. The Young-Yamanouchi basis is labelled by the standard Young tableaux $\lambda$ of shape $R$, and the transposition $(j,j+1)$ acts on them as
\begin{equation}
(j,j+1) \, \ket{R, \lambda} = \frac{1}{\rho_\lambda (j+1,j)} \, \ket{R, \lambda} + \sqrt{1 - \frac{1}{\rho_\lambda (j+1,j)^2} } \, \ket{R, (j,j+1) \lambda} \,,
\label{def:Ybasis}
\end{equation}
where $\rho_\lambda (j+1,j)$ is the axial distance from $j+1$ to $j$ in the standard Young tableau $\lambda$.\footnote{The axial distance between $a$ and $b$ is defined by counting the number of boxes we need to pass through from $a$ to $b$ in the Young tableau. We add $+1$ by going left or down, and $-1$ by going right or up. For example, $\rho_{\tiny \young(123,4)} (3,4) = -\rho_{\tiny \young(123,4)} (4,3) = +3$.}
The matrix representation of all other elements follows from \eqref{def:Ybasis}. It is straightforward to compute the branching coefficient from \eqref{branching via projector}.

The final ingredient to obtain the diagonal operators \eqref{def:OR,Lambda1,tau2} is the CG coefficient defined by \eqref{def:CG coeff}.
To extract $S^{\tau, \Lambda R R'}_{kij}$, we apply $\sigma \in S_{2n}$ to \eqref{def:CG coeff} and compare both sides.

Let us explain the computation of the CG coefficients in detail. We rewrite the matrix representation of the product $D^R \otimes D^R (\sigma)$ into a block diagonal form by a similarity transformation, 
\begin{equation}
(U_R^T)^{-1} \( D^R \otimes D^R (\sigma) \) U_R =
\begin{pmatrix}
D^{r_1} (\sigma) & & \\
& D^{r_2} (\sigma) & \\
& & \ddots
\end{pmatrix}.
\end{equation}
The CG coefficients are equal to the elements of the rotation matrix $U_R$\,. 
We compute each column of $U_R$ by using the Young symmetriser,
\begin{equation}
P_\lambda = \sum_{\sigma \in S_{2n}} p_\lambda (\sigma) \, \sigma
\equiv N_\lambda \prod_{k \in {\rm Column} (\lambda) } A_k \, \prod_{\ell \in {\rm Row} (\lambda) } S_\ell \,,
\label{def:Young symmetrizer}
\end{equation}
where $p_\lambda (\sigma)$ is a coefficient, $N_\lambda$ is a normalisation constant, and $A_k$ (or $S_\ell$) is anti-symmetric (or symmetric) combination of the entries in the $k$-th column (or $\ell$-th row) of the standard Young tableau $\lambda$, respectively. For example,
\begin{equation}
P_{\tiny \young(12,3)} = N_{\tiny \young(12,3)} \[ {\rm id} - (13) \] \[ {\rm id} + (12) \], \qquad
P_{\tiny \young(13,2)} = N_{\tiny \young(13,2)} \[ {\rm id} - (12) \] \[ {\rm id} + (13) \].
\label{YS type21}
\end{equation}
The Young symmetriser projects $D^{R \otimes R}$ onto the state corresponding to $\{ \tilde e_\lambda \}$.\footnote{The new basis $\{ \tilde e_\lambda \}$ is not orthogonal. The new basis is related to the Young-Yamanouchi basis in a trivial way after orthogonalization, in the simple cases studied here.}
Thus, the combination $\sum_\sigma p_{\lambda} (\sigma) \, D^{R \otimes R} (\sigma)$ becomes a rank-one matrix corresponding to a single eigenvector $\tilde e_\lambda$\,. By collecting all eigenvectors (and orthogonalizing them appropriately), we obtain the rotation matrix $U_R$\,.

\subsection{Explicit operators}

\subsubsection*{Length two}

The group $S_2$ has two irreducible representations, symmetric and anti-symmetric. 
Their tensor products decompose as
\begin{equation}
\rep{1,1} \otimes \rep{1,1} = \rep{2}, \qquad
\rep{2} \otimes \rep{2} = \rep{2},
\end{equation}
The symmetric representation $\rep{2}$ is an even partition. The representation matrices are
\begin{equation}
D^{\rep{2}} (\alpha) = 1, \qquad
D^{\rep{1,1}} (\alpha) = {\rm sign} \, (\alpha).
\end{equation}
The CG coefficients and the branching coefficients are trivial. 
The operators \eqref{def:OR,Lambda1,tau2} are given by
\begin{alignat}{9}
O^{\rep{2},\rep{2}} &= O_{\rm id} + O_{(12)}
& &= \tr ( \Phi^{a} ) \, \tr ( \Phi^{a} ) + \tr ( \Phi^{a} \Phi^{a} ),
\notag \\
O^{\rep{2},\rep{1,1}} &= O_{\rm id} - O_{(12)}
& &= \tr ( \Phi^{a} ) \, \tr ( \Phi^{a} ) - \tr ( \Phi^{a} \Phi^{a} ).
\end{alignat}

\subsubsection*{Length four}

We need the following irreducible representations of $S_4$
\begin{alignat}{9}
R &= \{ \rep{1^4}, \rep{2,1^2}, \rep{2^2}, \rep{3,1},  \rep{4} \}, &\qquad d_R &= \{ 1,3,2,3,1 \}, 
\notag \\
\Lambda_1 &= \{ \rep{2^2}, \rep{4} \}, &\qquad d_{\Lambda_1} &= \{ 2, 1 \}.
\end{alignat}
The tensor products decompose as
\begin{equation}
\begin{gathered}
\phantom{ }
\rep{1^4} \otimes \rep{1^4} = \rep{4}, \qquad
\rep{2^2} \otimes \rep{2^2} = \rep{1^4} \oplus \rep{2^2} \oplus \rep{4}, \qquad
\rep{4} \otimes \rep{4} = \rep{4},
\\[1mm]
\rep{2,1^2} \otimes \rep{2,1^2} = \rep{3,1} \otimes \rep{3,1}
= \rep{2,1^2} \oplus \rep{2^2} \oplus \rep{3,1} \oplus \rep{4}.
\end{gathered}
\end{equation}
This decomposition is multiplicity-free, so we can drop the index $\tau$ in $\cO^{R,\Lambda_1,\tau}$.

After the procedures of Appendix \ref{app:generality}, we obtain 
\begin{align}
O^{\rep{4},\rep{4}} &= \tr( a_1 )^2 \, \tr( a_2 )^2 + 2 \, \tr( a_1 )^2 \, \tr( a_2 a_2 ) + 4 \, \tr( a_1 ) \, \tr( a_2 ) \, \tr( a_1 a_2 )+8 \, \tr( a_1 ) \, \tr( a_1 a_2 a_2 )
\notag \\
&\quad +2 \, \tr( a_1 a_2 )^2+\, \tr( a_1 a_1 ) \, \tr( a_2 a_2 )+4 \, \tr( a_1 a_1 a_2 a_2 )+2 \, \tr( a_1 a_2 a_1 a_2 ) ,
\notag \\[1mm]
O^{\rep{1^4},\rep{4}} &= 
\tr( a_1 )^2 \, \tr( a_2 )^2 - 2 \, \tr( a_1 )^2 \, \tr( a_2 a_2 ) - 4 \, \tr( a_1 ) \, \tr( a_2 ) \, \tr( a_1 a_2 ) + 8 \, \tr( a_1 ) \, \tr( a_1 a_2 a_2 )
\notag \\
&\quad +2 \, \tr( a_1 a_2 )^2+\, \tr( a_1 a_1 ) \, \tr( a_2 a_2 ) - 4 \, \tr( a_1 a_1 a_2 a_2 ) - 2 \, \tr( a_1 a_2 a_1 a_2 ) ,
\notag \\[1mm]
O^{\rep{2^2},\rep{4}} &= - \frac{2}{\sqrt 2} \Big\{ \tr( a_1 )^2 \, \tr( a_2 )^2 - 4 \, \tr( a_1 ) \, \tr( a_1 a_2 a_2 )
+ \tr (a_1 a_1) \, \tr (a_2 a_2) + 2 \, \tr (a_1 a_2) \, \tr (a_1 a_2) \Big\},
\notag \\[1mm]
O^{\rep{2^2},\rep{2^2}} &= - \frac{4}{\sqrt 2} \, \Big\{ 
\tr( a_1 )^2 \, \tr( a_2 a_2 ) - \tr( a_1 ) \, \tr( a_2 ) \, \tr( a_1 a_2 )
+ \tr (a_1 a_2 a_1 a_2) - \tr (a_1 a_1 a_2 a_2 ) \Big\} ,
\notag \\[1mm]
O^{\rep{3,1},\rep{4}} &= \frac{1}{\sqrt 3} \Big\{ 3 \, \tr( a_1 )^2 \, \tr( a_2 )^2 + 2 \, \tr( a_1 )^2 \, \tr( a_2 a_2 )+4 \, \tr( a_1 ) \, \tr( a_2 ) \, \tr( a_1 a_2 )
\notag \\
&\quad -2 \, \tr( a_1 a_2 )^2 - \tr( a_1 a_1 ) \, \tr( a_2 a_2 )-4 \, \tr( a_1 a_1 a_2 a_2 )-2 \, \tr( a_1 a_2 a_1 a_2 )
\Big\},
\notag \\[1mm]
O^{\rep{2,1^2},\rep{4}} &= \frac{1}{\sqrt 3} \Big\{ 3 \, \tr( a_1 )^2 \, \tr( a_2 )^2 - 2 \, \tr( a_1 )^2 \, \tr( a_2 a_2 )
- 4 \, \tr( a_1 ) \, \tr( a_2 ) \, \tr( a_1 a_2 )
\notag \\
&\quad - 2 \, \tr( a_1 a_2 )^2 - \tr( a_1 a_1 ) \, \tr( a_2 a_2 ) + 4 \, \tr( a_1 a_1 a_2 a_2 ) + 2 \, \tr( a_1 a_2 a_1 a_2 )
\Big\},
\notag \\[1mm]
O^{\rep{3,1},\rep{2^2}} &= \frac{4}{\sqrt 6} \Big\{ \tr ( a_1 )^2 \,\tr ( a_2 a_2 ) - \tr ( a_1 ) \,\tr ( a_2 ) \,\tr ( a_1 a_2 )
\notag \\
&\quad
- \tr ( a_1 a_2 )^2 + \tr ( a_1 a_1 ) \,\tr ( a_2 a_2 ) + \tr ( a_1 a_1 a_2 a_2 ) - \tr ( a_1 a_2 a_1 a_2 ) 
\Big\},
\notag \\[1mm]
O^{\rep{2,1^2},\rep{2^2}} &= - \frac{4}{\sqrt 6} \Big\{ \tr ( a_1 )^2 \,\tr ( a_2 a_2 )-\tr ( a_1 ) \,\tr ( a_2 ) \,\tr ( a_1 a_2 )
\notag \\
&\quad
+ \tr ( a_1 a_2 )^2 - \tr ( a_1 a_1 ) \,\tr ( a_2 a_2 )+\tr ( a_1 a_1 a_2 a_2 )-\tr ( a_1 a_2 a_1 a_2 ) 
\Big\},
\label{diag ops length-four}
\end{align}
where we used the notation $\tr(a_1 a_2 a_3 a_4) = \tr (\Phi^{a_1} \Phi^{a_2} \Phi^{a_3} \Phi^{a_4})$. 
Their free two-point functions are given by \eqref{two-pt_rep_basis}.


\section{Relation to covariant approach}\label{sec:BHR}
\subsection{$O(N_f) \times S_{2n} $ CG coefficients}

In \cite{BHR08} a general construction of free-field diagonal operators was given, where the operators are built out of fields transforming in a general representation $V$ of a general global symmetry group $G$.
This general construction requires the explicit computation of the Clebsch-Gordan coefficients, which decompose the tensor products $ V^{ \otimes m }$ in terms of irreducible representations of $ G \times S_m $. 
Our approach to the free-field diagonal operators in the current paper is similar to \cite{BHR08} since both diagonal operators carry the same representation labels.
However, the two operators look slightly different, since the mesonic operators discussed here do not involve the CG coefficients. 

We will show that the two operators are identical, by giving an explicit formula for the relevant CG coefficients.
Recall that in \cite{BHR08} the colour and flavour indices are treated separately. A gauge-covariant operators turn into gauge-invariant operators by combining indices appropriately.
Instead, we may fix a gauge in the gauge-covariant form, and then combine them into gauge-invariant operators. In this way we can reproduce the mesonic operators.
The CG coefficients are related to the product of branching coefficients and Kronecker delta's.



\bigskip
Let us just focus on the CG problem for the $2n$-fold tensor product of 
$V_F$ with $G=O(N_f)$. In particular, we are interested in the one-dimensional representation of both $S_{2n}$ and $SO(N_f)$.  
Let $v_a$ be basis vectors in $V_F$. Consider 
   \bea 
   v_{ a_1} \otimes v_{ a_2} \otimes \cdots \otimes  v_{ a_{2n} } 
   \eea
   A vector $v_{ \rho}$, parametrised by a permutation $ \rho \in S_{ 2n }$ 
  which controls the pairwise contractions, is invariant under $ O ( N_f )$ : 
  \bea 
  v_{ \rho} = \delta^{ a_{ \rho(1)} a_{ \rho(2)} } \cdots \delta^{ a_{ \rho(2n-1)} a_{ \rho(2n )} }  ~~ v_{ a_1} \otimes \cdots \otimes v_{ a_{ 2n  } }  
\label{def:vrho}
  \eea
   Following a theme we have seen repeatedly, whenever we have some invariants parametrised by permutations, in the present case tensor products of vectors invariant 
    under $ O ( N_f )$, we must ask about the redundancy in the description. 
    Here the redundancy is 
    \bea 
    v_{ \rho } = v_{ \gamma \rho} 
    \eea
    for $ \gamma \in S_n [ S_2 ] $. 
    So we can also write 
    \bea 
    v_{ \rho} = { 1 \over 2^n n! } \sum_{ \gamma \in S_{n} [ S_2 ] } v_{ \gamma \rho  } 
    \eea
Again, following a familiar theme, disentangle these equivalence classes by using representation theory. 
The first step is to define 
    \bea\label{vproj} 
    v^{ \Lambda_1 }_{ IJ } = { 1 \over ( 2n ) ! } 
     \sum_{ \rho \in S_{ 2n } } D^{ \Lambda_1}_{ I J } ( \rho )  ~  v_{ \rho } 
    \eea
   Exploiting the invariance in the Fourier transformed basis,
   \bea 
   v^{ \Lambda_1 }_{ I J } && =  { 1 \over ( 2n ) ! } { 1 \over \abs{S_n[S_2]} } 
     \sum_{ \rho \in S_{ 2n } } D^{ \Lambda_1}_{ I J } ( \rho )  ~ \sum_{ \gamma \in S_n [ S_2 ] }  v_{ \gamma \rho } \cr 
      && ={ 1 \over ( 2n ) ! } { 1 \over \abs{S_n[S_2]} } \sum_{ \gamma \in S_n [ S_2 ] }
     \sum_{ \rho \in S_{ 2n } } D^{ \Lambda_1}_{ I J } (  \rho \gamma^{-1}  )  ~   v_{ \rho } \cr 
     && = { 1 \over ( 2n ) ! } 
     \sum_{ \rho \in S_{ 2n } } D^{ \Lambda_1}_{ I J } ( \rho  \, p_{1_{ S_n[ S_2 ] } }  ) ~  v_{  \rho }  \cr 
     &&   = { 1 \over ( 2n ) ! } 
     \sum_{ \rho \in S_{ 2n } }  D^{ \Lambda_1}_{ I K } ( \rho  ) B^{ \Lambda_1 \rightarrow 1 }_{ K  } B^{ \Lambda_1 \rightarrow 1 }_{ J }  v_{ \rho}  
     \eea
    The branching coefficient is zero if $ \Lambda_1 $ is not even. 
The invariant vectors are therefore\footnote{See e.g. Appendix B of \cite{BHR07} of the relevant fact from linear algebra.}
    \bea 
    v^{ \Lambda_1}_{ I } \equiv { 1 \over ( 2n ) ! } \, 
     \sum_{ \rho \in S_{ 2n } }  D^{ \Lambda_1}_{ I K } ( \rho  ) B^{ \Lambda_1 \rightarrow 1 }_{ K  } ~ v_{ \rho} 
\label{def:inv-vec-v}    \eea
Recalling  the definition of $ v_{\rho}$
    \bea 
      v^{ \Lambda_1}_{ I } = { 1 \over ( 2n ) ! } 
     \sum_{ \rho \in S_{ 2n } }  D^{ \Lambda_1}_{ I K } ( \rho  ) B^{ \Lambda_1 \rightarrow 1 }_{ K  } ~ \delta^{ a_{ \rho(1)} a_{ \rho(2)} } \cdots \delta^{ a_{ \rho(2n-1)} a_{ \rho(2n )} }  ~~ v_{ a_1} \otimes \cdots \otimes v_{ a_{ 2n  } }
    \eea
This can be written in terms of a CG coefficient coupling $V_F^{ \otimes 2n }$ to the state $(\Lambda_1 , I) \times \emptyset$ of $ S_{ 2n } \times O ( N_f )$,
where $\Lambda_1$ refers to the representation of $S_{2n}$,
$I$ is the state label of $\Lambda_1$,
and $\emptyset$ the one-dimensional representation of $O(N_f)$:
   \bea 
   v^{ \Lambda_1}_{ I } = 
C^{ \vec a }_{ \Lambda_1, I }  v_{ a_1} \otimes \cdots \otimes v_{ a_{ 2n  } }
   \eea
Multiplicity labels are not needed.

It is tempting to identify the CG coefficients as 
\bea 
C^{ \vec a }_{ \Lambda_1, I }   =  { 1 \over ( 2n ) ! } \, 
     \sum_{ \rho \in S_{ 2n } }  D^{ \Lambda_1}_{ I K } ( \rho  ) B^{ \Lambda_1 \rightarrow 1 }_{ K  } ~ \delta^{ a_{ \rho(1)} a_{ \rho(2)} } \cdots \delta^{ a_{ \rho(2n-1)} a_{ \rho(2n )} } 
\label{BHR-Clebsch}
\eea
There is an important subtlety of normalisation which has to be considered when comparing  to \cite{BHR08}. The key point is that the normalisation of the above-defined    $C^{ \vec a }_{ \Lambda_1, I }  $ is 
\begin{equation}
\sum_{\vec a} C^{ \vec a }_{ \Lambda_1, I, \tau_{\Lambda} } C^{ \vec a }_{ \Lambda'_1, I', \tau'_{\Lambda} } = N_{CG} \, \delta_{\Lambda_1 \Lambda'_1} \delta_{I I'} \delta_{\tau_{\Lambda} \tau'_{\Lambda}} 
\end{equation}
where $N_{CG}$ is determined below.
Let us compute $ v^{ \Lambda}_{I} v^{ \Lambda}_{ J }  $ (no sum over $\Lambda_1$)
\begin{align}
\langle v^{ \Lambda_1 }_I , v^{ \Lambda_1 }_J \rangle 
&= { 1 \over ( 2n)!^2} \,
\sum_{\vec a, \vec b} C^{ \vec a }_{ \Lambda_1, I } C^{ \vec b }_{ \Lambda_1, J } 
\bigl\langle v_{\vec a} , v_{\vec b} \bigr\rangle
\\
&= { 1 \over  ( 2n)!^2}  \sum_{ \rho_1 } \sum_{ \rho_2} 
D^{ \Lambda_1}_{ I K } ( \rho_1 ) D^{ \Lambda_1}_{ J L } ( \rho_2 )
B^{ \Lambda_1 \rightarrow 1 }_{ K  } B^{ \Lambda_1 \rightarrow 1 }_{ L  } 
\langle v_{ \rho_1 } , v_{ \rho_2 } \rangle .
\end{align} 
Since $C^{ \vec a }_{ \Lambda_1, I }  = C^{ \sigma(\vec a) }_{ \Lambda_1, I } $ for any $\sigma \in S_{2n}$, the first line is also written as
\begin{equation}
\langle v^{ \Lambda_1 }_I , v^{ \Lambda_1 }_J \rangle =  { N_{CG} \over  ( 2n)!^2}  \, \delta_{IJ} \,.
\label{compare NCG}
\end{equation}
Note we assumed, as in \cite{BHR08}, that the vectors in $V_F $ are unit normalised
\bea 
\langle v_a , v_b \rangle = \delta_{ab} \,, \qquad
\langle v_{\vec a} , v_{\vec b} \rangle = \prod_{k=1}^n \delta_{a_k b_k} \,. 
\eea
The norm of the permutation-parametrised vectors is 
\begin{align}
\langle v_{ \rho_1} , v_{ \rho_2} \rangle 
&= 
\delta^{ a_{ \rho_1 (1)} a_{ \rho_1 (2)} } \cdots \delta^{ a_{ \rho_1 (2n-1)}  a_{ \rho_1 (2n)}  }
\delta^{ b_{ \rho_2 (1)} b_{ \rho_2 (2)} } \cdots \delta^{ b_{ \rho_2 (2n-1)} b_{ \rho_2 (2n)}  } \,
\langle v_{ a_1} \cdots v_{a_{2n}} , v_{ b_1} \cdots v_{ b_{2n}} \rangle 
\notag \\
&=  W ( \rho_2^{-1} \rho_1 ) 
\end{align}
This determines $N_{CG}$ via
\begin{align}
\langle v^{ \Lambda_1 }_I , v^{ \Lambda_1 }_{J} \rangle
& = {  \delta_{IJ}  \over ( 2n ) ! \, d_{\Lambda_1} } \sum_{ \rho  } 
W ( \rho )   \chi_{ \Lambda_1} (  \rho^{-1} \, p_{1_{S_n[S_2]}} ) 
= \delta_{IJ}  \,  { N_f^n \omega_{\Lambda_1/2} (\Omega_{2n}^{(f)}) \over (2n)! \, d_{\Lambda_1} } \,.
\end{align}

If we substitute the CG coefficients \eqref{BHR-Clebsch} to the general diagonal operator of \cite{BHR08},
we obtain gauge-invariant operators involving a sum over permutations 
$ \rho , \alpha$. After doing the sum over $ \rho$, our diagonal operators in section \ref{sec: rep basis} can be recovered.

\subsection{Baryonic operators}\label{sec:baryonic}

We explain how to count $SO(N_f)$ singlets following \cite{BHR08}.
Let us take the flavour part of the Schur-Weyl duality
\begin{align}
V_F^{\otimes 2n}=\bigoplus_{\substack{\Lambda_1 \\ c_1(\Lambda_1)\le N_f}}
\left(V_{\Lambda_1}^{GL(N_f)} \otimes V_{\Lambda_1}^{S_{2n}}\right).
\label{SWdual_flavour}
\end{align} 
We restrict $GL(N_f)$ to $O(N_f)$, and further to $SO(N_f)$ by the projection $\pi$ \cite{KoikeTerada,CGrood},\footnote{The unitary irreducible representations of $O(N_f)$ should satisfy $c_1(\Lambda_2)+c_2(\Lambda_2)\le N_f$.}
\begin{gather}
V_{\Lambda_1}^{GL(N_f)}
= \bigoplus_{\Lambda_2} V_{\Lambda_1,\Lambda_2}\otimes V_{\Lambda_2}^{O(N_f)}
= \bigoplus_{\Lambda_2} \otimes V_{\pi(\Lambda_2)}^{SO(N_f)} \,,
\\
\dim V_{\Lambda_1, \Lambda_2} = \sum_{\beta: {\rm even}} g(\Lambda_2,\beta;\Lambda_1),
\end{gather}
where $g(A,B;C)$ is the LR coefficient \eqref{def:LR coeff}.
The singlet representations of $SO(N_f)$ have two origins.
The first origin is an $O(N_f)$ singlet. 
The other is a non-singlet of $O(N_f)$ projected by $\pi$.
An example is $\pi([1^{N_f}]) = \emptyset$, corresponding to 
\begin{align}
\Phi_{[a_1}\cdots \Phi_{a_{N_f}]}=\frac{1}{N_f!}\epsilon_{a_1\cdots a_{N_f}}
\epsilon^{b_1\cdots b_{N_f}}\Phi_{[b_1}\cdots \Phi_{b_{N_f}]}. 
\end{align} 
Following the arguments in section \ref{sec:counting_SW}, the number of $SO(N_f)$ singlet operators is counted by 
\begin{equation}
\sum_{
\substack{ R \\c_1(R)\le N_c} }
\sum_{
\substack{ \Lambda_1 \\c_1(\Lambda_1)\le N_f} }
\sum_{p=0}^n
\sum_{\beta \vdash 2p} 
\sum_{ \Lambda_2 \vdash (2n-2p)}
C(R,R,\Lambda_1) g(\Lambda_2,\beta;\Lambda_1) \,
\delta_{\pi(\Lambda_2), \emptyset} \,.
\end{equation}
The mesonic operators are counted by setting $\Lambda_2=\emptyset$ in the above formula, yielding \eqref{SW counting}. The baryonic operators correspond to $\Lambda_2 \neq \emptyset$.


\section{Mixing matrix in detail}\label{sec:mixing_detail}
\setcounter{equation}{0}

In this Appendix we derive the mixing matrix on the permutation basis 
(\ref{mixing_matrix_mesonic_singlet}).

For our convenience we call each term of the following dilatation operator $H_i$
\begin{align}
H=
H_1+H_2=
-\frac{1}{2}tr[\Phi_m,\Phi_n][\check{\Phi}^m,\check{\Phi}^n]
-\frac{1}{4}tr[\Phi_m,\check{\Phi}^n][\Phi_m,\check{\Phi}^n]
\end{align} 
where\footnote{The Hamiltonian of integrable $SO(N_f)$ spin chain is obtained by changing the coefficient of $H_2$ to $-1/(N_f-2)$ and taking the planar limit \cite{MZ02}.}
\begin{align}
H_2=
-\frac{1}{4}
tr[\Phi_m,\check{\Phi}^n][\Phi_m,\check{\Phi}^n]
=&
-\frac{1}{2}
tr(\Phi_m \check{\Phi}^n \Phi_m\check{\Phi}^n)
+\frac{1}{2}tr(\Phi_m \Phi_m \check{\Phi}^n \check{\Phi}^n)
\nonumber \\
=&H_{21}+H_{22}.
\end{align} 
The following formulae are useful, 
\begin{align}
tr([\Phi_m,\Phi_n][\check{\Phi}^m,\check{\Phi}^n])
(\Phi_a)_{ij}(\Phi_b)_{kl}
&=2([\Phi_a,\Phi_b]_{kj}\delta_{il}-[\Phi_a,\Phi_b]_{il}\delta_{kj})
\label{act_first_dilatation} \\[1mm]
tr(\Phi_m \check{\Phi}^n \Phi_m\check{\Phi}^n)
(\Phi_a)_{ij}(\Phi_b)_{kl}
&=2\delta_{ab} (\Phi_m)_{il}(\Phi_m)_{kj}
\label{act_second_1_dilatation} \\[1mm]
tr(\Phi_m  \Phi_m \check{\Phi}^n\check{\Phi}^n)
(\Phi_a)_{ij}(\Phi_b)_{kl}
&=\delta_{ab} (\Phi_m \Phi_m)_{il}\delta_{jk}+\delta_{ab} (\Phi_m \Phi_m)_{kj}\delta_{il}.
\label{act_second_2_dilatation}
\end{align}

It is convenient to consider the dilatation operator acting on  
general operators built from $SO(N_f)$ scalars in \eqref{abbr trPhi}.
The action of $H_1$ is given by 
\begin{align}
&tr([\Phi_m,\Phi_n][\check{\Phi}^m,\check{\Phi}^n])
\, tr_{2n} (\sigma \Phi_{\vec{a}})
\nonumber \\
&=2 
\sum_{i\ne j}
tr_{2n} ([\sigma,(ij)] \Phi_{a_1}\otimes \cdots \otimes 
[\Phi_{a_i},\Phi_{a_j}] \otimes \cdots  \otimes 1 \otimes \cdots \otimes \Phi_{a_{2n}})
\nonumber \\
&=2
\sum_{i\ne j}
\sum_{\alpha \in S_{2n}}
\delta_{2n}([\sigma,(ij)]\alpha^{-1})
tr_{2n} (\alpha \Phi_{a_1}\otimes \cdots \otimes 
[\Phi_{a_i},\Phi_{a_j}] \otimes \cdots  \otimes 1 \otimes \cdots \otimes \Phi_{a_{2n}})
\label{action_H_1}
\end{align} 
where $[\Phi_{a_i},\Phi_{a_j}]$ is in the $i$-th slot and 
$1$ is in the $j$-th slot.\footnote{In the planar limit, only the terms $j=\sigma(i)$ and $j=\sigma^{-1}(i)$ survive.}
Here the sum $\sum_{i\ne j}$ is over 
different pairs ($i,j$), i.e. we do not distinguish $(i,j)=(1,2)$ and $(2,1)$.

In order to express the above operator in 
terms of (\ref{abbr trPhi}), 
we consider 
the decomposition $S_{2n} \rightarrow S_{2n-1}\times S_1 $ \cite{0701066,1012.3884}.
Elements in $S_{2n}$ can be expressed in terms of elements in 
$S_{2n-1}$ as
\begin{align}
\{\alpha \hspace{0.1cm}| \hspace{0.1cm}\alpha \in S_{2n} \} = 
\{\beta \,|\, \beta\in S_{2n-1}^{\langle j \rangle} \} \cup \{ \beta(jk) 
\hspace{0.1cm}| \hspace{0.1cm}
k=1,2,j-1,j+1,\cdots,2n ; \hspace{0.1cm}\beta\in S_{2n-1}^{\langle j \rangle}
\},
\label{Sn-1 split}
\end{align} 
where 
$S_{2n-1}^{\langle j \rangle}$ is 
the subgroup 
obtained by removing 
the $j$-th slot from $S_{2n}$.  
We illustrate how it works 
for the case $2n=3$, $j=3$. 
Take $(i,k)=(2,1)$. When $\alpha = \beta$,
\begin{equation}
tr_3 (\beta 
\Phi_{a_1}\otimes [\Phi_{a_2},\Phi_{a_3}] \otimes 1 )
=
N_c \, tr_2 (\beta 
\Phi_{a_1}\otimes [\Phi_{a_2},\Phi_{a_3}])
=
N_c \, tr_3 ([(23),\beta] 
\Phi_{a_1}\otimes \Phi_{a_2} \otimes \Phi_{a_3}) 
\end{equation}
and when $\alpha = \beta (jk)$,
\begin{equation}
tr_3 (\beta (31) 
\Phi_{a_1}\otimes [\Phi_{a_2},\Phi_{a_3}] \otimes 1 )
=
tr_2 (\beta 
\Phi_{a_1}\otimes [\Phi_{a_2},\Phi_{a_3}])
=
tr_3 ([(23),\beta] 
\Phi_{a_1}\otimes \Phi_{a_2} \otimes \Phi_{a_3})
\end{equation}
The case $(i,k)=(2,2)$ is same as above.

We then find that 
(\ref{action_H_1}) can be written as 
\begin{align}
& 
2N_c \sum_{i\ne j}
\sum_{\beta \in S_{2n-1}^{\langle j\rangle}}
\delta_{2n}([\sigma,(ij)]\beta^{-1})
tr_{2n} ([(ij),\beta] \Phi_{\vec{a}})
\nonumber \\
&
+2\sum_{i\ne j}
\sum_{k (k\neq j)}
\sum_{\beta \in S_{2n-1}^{\langle j \rangle}}
\delta_{2n}([\sigma,(ij)](jk)\beta^{-1})
tr_{2n} ([(ij),\beta]\Phi_{\vec{a}})
\nonumber \\
&=
2 \sum_{i\ne j}
\sum_{\beta \in S_{2n-1}^{\langle j\rangle}}
\delta_{2n}([\sigma,(ij)]X^{(j)}\beta^{-1})
tr_{2n} ([(ij),\beta] \Phi_{\vec{a}}),
\end{align} 
where 
we have introduced 
\begin{align}
X^{(j)}=N_c+\sum_{k (\neq j)} (kj).
\label{def:X(j)}
\end{align}

\quad 

Next study the action of $H_{21}$,
\begin{align}
&tr(\Phi_m \check{\Phi}^n \Phi_m\check{\Phi}^n)
\, tr_{2n} (\sigma \Phi_{\vec{a}})
= 2\sum_{i\ne j}
\delta_{a_i,a_j}
tr_{2n} ( (ij) 
\sigma \Phi_{a_1}\otimes \cdots \otimes 
\Phi_m \otimes \cdots  \otimes \Phi_m \otimes \cdots \otimes \Phi_{a_{2n}})
\end{align} 
where two $\Phi_m$'s are in the $i$-th position and the $j$-th position. 
Introducing the flavour contraction operator 
acting on two $\Phi$'s at $(i,j)$, 
\begin{align}
C_{(ij)}^f \Phi_a \otimes \Phi_{b}=\delta_{ab} \Phi_m \otimes \Phi_{m}, 
\end{align}
we have 
\begin{eqnarray}
tr(\Phi_m \check{\Phi}^n \Phi_m\check{\Phi}^n)
\, tr_{2n} (\sigma \Phi_{\vec{a}})
= 2\sum_{i\ne j}
C_{(ij)}^f
tr_{2n} ( (ij) 
\sigma 
\Phi_{\vec{a}}).
\end{eqnarray}

\quad

The action of $H_{22}$ is computed using (\ref{act_second_2_dilatation}) as 
\begin{align}
&
tr(\Phi_m  \Phi_m \check{\Phi}^n\check{\Phi}^n)
\, tr_{2n} (\sigma \Phi_{\vec{a}})
\nonumber \\
&=
2 \sum_{i\neq j}
\delta_{a_i a_j}
tr_{2n}((ij)\sigma \Phi_{a_1}\otimes \cdots \otimes \Phi_{m}\Phi_{m} \otimes \cdots 
\cdots \otimes 1\otimes \cdots \Phi_{a_{2n}})
\nonumber \\
&=
2 \sum_{i\neq j}
\delta_{a_i a_j}
\sum_{\alpha \in S_{2n}}
\delta_{2n}((ij)\sigma \alpha^{-1})
S^{(ij)}(\alpha),
\label{H_22_action_an_expression} \\[1mm]
S^{(ij)}(\alpha)
&:=tr_{2n}(\alpha \Phi_{a_1}\otimes \cdots \otimes \underbrace{\Phi_{m}\Phi_{m}}_i \otimes \cdots 
\cdots \otimes \underbrace{1}_j \otimes \cdots \Phi_{a_{2n}})
\end{align} 
where $i,j$ represent the site of $\Phi_m\Phi_m$ and that of $1$.
We now apply the reduction $S_{2n} \rightarrow S_{2n-1}^{\langle j \rangle }\times S_1$
to $S^{(ij)}(\alpha)$.
For example for $i=2,j=1$, we have
\begin{align}
&
S^{(21)}(\beta)
=
N_c \, tr_{2n-1} (
\beta
\Phi_m \Phi_m \otimes 
\Phi_{a_3}\otimes  \cdots \otimes \Phi_{a_{2n}})
\nonumber \\
&
S^{(21)}(\beta (1k))
=
tr_{2n-1} (
\beta
\Phi_m \Phi_m \otimes 
\Phi_{a_3}\otimes  \cdots \otimes \Phi_{a_{2n}})
\end{align} 
where $\beta \in S_{2n-1}^{\langle 1 \rangle}$,
and we can use the following formula,
\begin{align}
tr_{2n-1} (
\beta
\Phi_m \Phi_m \otimes 
\Phi_{a_3}\otimes  \cdots \otimes \Phi_{a_{2n}})
=
tr_{2n} (
 (12)\beta
\Phi_m \otimes \Phi_m \otimes 
\Phi_{a_3}\otimes  \cdots \otimes \Phi_{a_{2n}}). 
\end{align} 
We then find that 
(\ref{H_22_action_an_expression}) can be expressed by
\begin{align}
&tr(\Phi_m  \Phi_m \check{\Phi}^n\check{\Phi}^n)
\, tr_{2n} (\sigma \Phi_{\vec{a}})
= 2 \sum_{i\neq j} C_{(ij)}
\sum_{\beta \in S_{2n-1}^{\langle j \rangle}}
\delta_{2n}((ij)\sigma X^{(j)}\beta^{-1})
tr_{2n}((ij)\beta \Phi_{\vec{a}}).
\end{align}

Collecting these results,
\begin{align} 
&H_1 tr_{2n} (\sigma \Phi_{\vec{a}})
=
- \sum_{i\ne j}
\sum_{\beta \in S_{2n-1}^{\langle j\rangle}}
\delta_{2n}([\sigma,(ij)]X^{(j)}\beta^{-1})
tr_{2n} ([(ij),\beta] \Phi_{\vec{a}}),
\nonumber \\
&H_{21} tr_{2n} (\sigma \Phi_{\vec{a}})
=
-\sum_{i\ne j}
C_{(ij)}
tr_{2n} ( (ij) 
\sigma \Phi_{\vec{a}}),
\nonumber \\
&H_{22} tr_{2n} (\sigma \Phi_{\vec{a}})
=
 \sum_{i\neq j}
C_{(ij)}
\sum_{\beta \in S_{2n-1}^{\langle j \rangle}}
\delta_{2n}((ij)\sigma X^{(j)}\beta^{-1})
tr_{2n}((ij)\beta \Phi_{\vec{a}}).
\label{H1-mixing-formula0}
\end{align} 
Note that 
the $N_c$-dependence appears only in $X^{(j)}$.

\quad 

Let us next focus on the mesonic singlet operators
\begin{align} 
\cO(\sigma) \equiv \cO_{\sigma}=tr_{2n} (\sigma \Phi_{a_1}\otimes \Phi_{a_1}
\otimes \Phi_{a_2}\otimes \Phi_{a_2}
\cdots \otimes \Phi_{a_{n}}\otimes \Phi_{a_{n}}).
\end{align} 
We decompose the sum over $i,j$ into 
\begin{align} 
\sum_{i\neq j}=\sum_{(i,j)}+\sum_{\langle i,j\rangle}
\end{align} 
where $(i,j)$ run over $(1,2),(3,4),\cdots$, and 
$\langle i,j\rangle$ over the other pairs.\footnote{$i = \Sigma_0 (j)$ in the sum over $(i,j)$.} 
We have 
\begin{align}\label{H1-mixing-formula}  
&H_1 \cO(\sigma)
=
- \sum_{\langle i,j\rangle}
\sum_{\beta \in S_{2n-1}^{\langle j\rangle}}
\delta_{2n}([\sigma,(ij)]X^{(j)}\beta^{-1})
\cO([(ij),\beta]) ,
\nonumber \\
&H_{21} \cO(\sigma)
=
-N_f \sum_{(i,j)}
\cO( (ij) \sigma )
-\sum_{\langle i,j\rangle}
C_{(ij)}
\cO( (ij) \sigma ),
\nonumber \\
&H_{22} \cO(\sigma)
=
 N_f \sum_{(i,j)}
\sum_{\beta \in S_{2n-1}^{\langle j \rangle}}
\delta_{2n}((ij)\sigma X^{(j)}\beta^{-1})
\cO((ij)\beta )
\nonumber \\
& \hspace{2cm}
+
 \sum_{\langle i,j\rangle}
C_{(ij)}
\sum_{\beta \in S_{2n-1}^{\langle j \rangle}}
\delta_n((ij)\sigma X^{(j)}\beta^{-1})
\cO((ij)\beta ).
\end{align} 
The mixing matrix 
\begin{align} 
H_i \cO(\sigma)=\sum_{\tau\in S_{2n}}M_{\sigma,\tau}^{(i)}O(\tau)
\end{align} 
is given by 
\begin{align}\label{MixM1} 
M^{(1)}_{\sigma,\tau}
=&
- \sum_{\langle i,j\rangle}
\sum_{\beta \in S_{2n-1}^{\langle j\rangle}}
\delta_{2n}([[\sigma],(ij)]X^{(j)}\beta^{-1})
\delta_{2n}([[\tau^{-1}],(ij)]\beta) ,
\nonumber\\
=&- 2n \left(n-1\right)
\sum_{\beta \in S_{2n-1}^{\langle 2n \rangle}}
\delta_n([[\sigma],(1,2n)]X^{(2n)}\beta^{-1})
\delta_n([[\tau^{-1}],(1,2n)]\beta) ,
\\[1mm]
M^{(21)}_{\sigma,\tau}
=&
-N_f \sum_{(i,j)}\delta_{2n}( [\tau^{-1}](ij) [\sigma] )
-\sum_{\langle i,j\rangle}
\delta_{2n}\left([\tau^{-1}] (\Sigma_0(i)j)(ij) [\sigma] (\Sigma_0(i)j)\right),
\nonumber \\
=&
-N_f n\delta_{2n}( [\tau^{-1}](12) [\sigma] )
-2n\left(n-1\right)
\delta_{2n}\left([\tau^{-1}] (2,2n)(1,2n) [\sigma] (2,2n)\right),
\label{MixM21}  
\end{align} 
where $\Sigma_0=(12)(34)\cdots (2n-1,2n)$, 
\begin{align}\label{MixM22} 
M^{(22)}_{\sigma,\tau}
=&\
 N_f \sum_{(i,j)}
\sum_{\beta \in S_{2n-1}^{\langle j \rangle}}
\delta_{2n}((ij)[\sigma] X^{(j)}\beta^{-1})
\delta_{2n}( [\tau^{-1}](ij)\beta )
\nonumber \\
&
+
 \sum_{\langle i,j\rangle}
\sum_{\beta \in S_{2n-1}^{\langle j \rangle}}
\delta_{2n}((ij)[\sigma ]X^{(j)}\beta^{-1})
\delta_{2n}([\tau^{-1}](\Sigma_0(i)j)(ij)\beta (\Sigma_0(i)j))
\nonumber \\
=&\
 N_f n
\sum_{\beta \in S_{2n-1}^{\langle 2n \rangle}}
\delta_{2n}((2n-1,2n)[\sigma] X^{(n)}\beta^{-1})
\delta_{2n}( [\tau^{-1}](n-1,n)\beta )
\nonumber \\
&
+
2n\left(n-1\right)
\sum_{\beta \in S_{2n-1}^{\langle n \rangle}}
\delta_{2n}((1,2n)[\sigma] X^{(2n)}\beta^{-1})
\delta_{2n}([\tau^{-1}](2,2n)(1,2n)\beta (2,2n)).
\end{align}


\end{document}